\documentclass{lmcs} %%% last changed 2014-08-20
\pdfoutput=1
\usepackage[utf8]{inputenc}

\usepackage{lastpage}
\lmcsdoi{16}{2}{8}
\lmcsheading{}{\pageref{LastPage}}{}{}%
{Feb.~01,~2019}{Jun.~02,~2020}{}

\keywords{Combinatory logic, B combinator, Lambda calculus}

\usepackage{forest}
\tikzset{
  solid node/.style={circle, draw, inner sep=1.0, fill=black},
}
\forestset{
  Nd/.style={solid node, label={left:#1}, minimum size=4pt}
}
\usepackage{booktabs} % for \toprule, \midrule, \bottomrule
\usepackage{subcaption}
\newcommand\ie{i.e.}

\DeclareRobustCommand\dcolonequals{\mathbin{::=}}
\DeclareRobustCommand\dplus{\mathbin{+\!\!\!+}}

\newcommand\labeqn[1]{\label{eqn:#1}}
\newcommand\refeqn[1]{(\ref{eqn:#1})}
\newcommand\labfig[1]{\label{fig:#1}}
\newcommand\reffig[1]{\textrm{Fig.\,\ref{fig:#1}}}
\newcommand\Reffig[1]{\textrm{Figure~\ref{fig:#1}}}
\newcommand\labtbl[1]{\label{tbl:#1}}

\newcommand\Reftbl[1]{\textrm{Table~\ref{tbl:#1}}}
\newcommand\labthr[1]{\label{thr:#1}}
\newcommand\refthr[1]{\textrm{Theorem~\ref{thr:#1}}}

\newcommand\lablem[1]{\label{lem:#1}}
\newcommand\reflem[1]{\textrm{Lemma~\ref{lem:#1}}}

\newcommand\labcnj[1]{\label{cnj:#1}}
\newcommand\refcnj[1]{\textrm{Conjecture~\ref{cnj:#1}}}
\newcommand\labsec[1]{\label{sec:#1}}
\newcommand\refsec[1]{\textrm{Section~\ref{sec:#1}}}

\newcommand\labalg[1]{\label{alg:#1}}
\newcommand\refalg[1]{\textrm{Algorithm~\ref{alg:#1}}}

\newcommand\paren[1]{\left(#1\right)}
\newcommand\sapp[2]{#1_{\paren{#2}}}

\newcommand\UL[1]{\underline{#1}}
\newcommand\WL[1]{\underline{\underline{#1}}}

\newcommand\CL[1]{\mathbf{CL}(#1)}

\newcommand\nodes{\mathcal{L}}

\newcommand\x{\star}
\def\<#1>{\langle #1\rangle}
\newcommand\poly[1]{%
\langle\dots\langle\langle\x,\underbrace{\x\rangle,\x\rangle,\dots,\x\rangle}_{#1}%
}
\newcommand\size[1]{\left|\!\left|#1\right|\!\right|}
\newcommand\concat{\mathbin{{+}\mspace{-8mu}{+}}}

\newcommand\Proof{\textit{Proof}}
\newcommand\lampat[2]{#1#2}
\makeatletter
\newcommand\Lam[2]{\lambda\Lam@main{\lampat{#1}{#2}}}
\newcommand\Lam@main[1]{\@ifnextchar\bgroup{\Lam@more{#1}}{{#1}}}
\newcommand\Lam@more[2]{\Lam@main{\lampat{#1}{#2}}}
\newcommand\dotsLam[1]{\lampat{\dots}{#1}}
\newcommand\dotLam[1]{\lampat{}{#1}}
\makeatother

\usepackage{textcomp}
\newcommand\ctilde{\text{\raisebox{.6ex}{\texttildelow}}}

\theoremstyle{defC}\newtheorem{algoC}[thm]{Algorithm} % algorithm without ( )
\theoremstyle{plain}\newtheorem{claim}[thm]{Claim}
\theoremstyle{plain} %\crefname{satz}{Satz}{S\"atze}
\def\eg{{\em e.g.}}

\begin{document}

\title{On properties of $B$-terms}
%\title{On an interesting property of \(B\)-terms}
%\title[Instructions]{Instructions for Authors\\How to prepare papers
%  for LMCS using \texorpdfstring{\MakeLowercase{\texttt{lmcs.cls}}}{lmcs.cls}\rsuper*\\Version of 
%  2014-02-01}
%\titlecomment{{\lsuper*}OPTIONAL comment concerning the title, \eg, 
%  if a variant or an extended abstract of the paper has appeared elsewhere.}

\author[M.~Ikebuchi]{Mirai Ikebuchi}	%required
\address{Massachusetts Institute of Technology, Cambridge, MA, USA}	%required
\email{ikebuchi@mit.edu}

\author[K.~Nakano]{Keisuke Nakano}	%required
\address{Tohoku University, Sendai, Miyagi, Japan}	%required
\email{k.nakano@acm.org}  %optional
%\thanks{thanks 1, optional.}	%optional

%\author[B.~Name2]{Bob Name2}	%optional
%\address{address2; addresses should initially be duplicated, even if
%  authors share an affiliation}	%optional
%\email{name2@email2; ditto for email addresses}  %optional
%\thanks{thanks 2, optional.}	%optional
%
%\author[C.~Name3]{Carla Name3}	%optional
%\address{address 3}	%optional
%\urladdr{name3@url3\quad\rm{(optionally, a web-page can be specified)}}  %optional
%\thanks{thanks 3, optional.}	%optional

%% etc.

%% required for running head on odd and even pages, use suitable
%% abbreviations in case of long titles and many authors:

%%%%%%%%%%%%%%%%%%%%%%%%%%%%%%%%%%%%%%%%%%%%%%%%%%%%%%%%%%%%%%%%%%%%%%%%%%%

%% the abstract has to PRECEDE the command \maketitle:
%% be sure not to issue the \maketitle command twice!

%!TEX root = main.tex
\begin{abstract}
\(B\)-terms are built from the \(B\) combinator alone
defined by \(B\equiv\Lam{f}{g}{x}. f(g~x)\), which is
well known as a function composition operator.
This paper investigates an interesting property of \(B\)-terms,
that is, whether repetitive right applications of a \(B\)-term cycles or not.
We discuss conditions for \(B\)-terms to have and not to have the property
through
a sound and complete equational axiomatization.
Specifically, 
we give examples of \(B\)-terms which have the cyclic property 
and show that there are infinitely many \(B\)-terms which do not have the property.
Also, we introduce another interesting property 
about a canonical representation of \(B\)-terms
that is useful to detect cycles, or equivalently, to prove the cyclic property, 
with an efficient algorithm.
\end{abstract}

%%% Local Variables:
%%% mode: latex
%%% TeX-master: "main.tex"
%%% End:

\maketitle

%% start the paper here:
% Introduction
%!TEX root = main.tex
\section*{Introduction}
\labsec{intro}
The `bluebird' combinator \(B=\Lam{f}{g}{x}.~f(g~x)$
is well known~\cite{Schoenfinkel24ma, Curry30ajm1,Smullyan12book}
as a bracketing combinator or composition operator,
which has a principal type 
$(\alpha\to\beta)\to(\gamma\to\alpha)\to\gamma\to\beta$.
A function $B~f~g$ (also written as $f\circ g$)
% synthesized from two functions $f$ and $g$
takes a single argument $x$ % to apply $g$ and
and returns the term $f(g~x)$.
In the general case that \(g\) takes \(n\) arguments,
the composition can be given by
\(\lambda{x_1}\dotsLam{x_n}.\, f(g~x_1~\dots~x_n)\).
We call it the \emph{\(n\)-argument composition} of \(f\) and \(g\).
% the synthesized function is defined by
%$\lambda x_1.\cdots\lambda x_n. f~(g~x_1~\dots~x_n)$,
Interestingly, the function can be given as $B^n\, f~ g$
% $\cnum{n}~B$
% where $e^n$ is an $n$-fold composition of function $e$ such that
% $e^0 = \lambda x. x$ and $e^{n+1}=B~e^{n}~e$ for $n\geq0$.
where $e^n$ stands for the $n$-fold composition $\underbrace{e \circ \dots \circ e}_n$ of the function $e$,
or equivalently defined by \(e^n\, x = \underbrace{e~(\dots(e}_n~x))\).
%~\cite[Definition 2.1.9]{Barendregt84book}.
% $\cnum{n}$ is a Church's numeral for an integer $n$ such that
%$\cnum{0}~f = \lambda x. x$ and $\cnum{n}~f = B~(\cnum{n-1}~f)~f$ for $n\geq 0$.
%We call {\em $n$-argument composition\/} 
%for the generalized composition represented by $B^n$.
%$\cnum{n}~B$.
This fact can be shown by an easy induction.

Now we consider the $2$-argument composition expressed as 
\(B^2=\Lam{f}{g}{x}{y}.\,f(g~x~y)\).
%$\cnum{2}~B$. 
From the definition, we have $B^2 = B\circ B = B~B~B$.
%$\cnum{2}~B = B~(\cnum{1}~B)~B = B~B~B$
Note that function application
is considered left-associative, that is, $f~a~b=(f~a)~b$.
Thus $B^2$ is expressed as a term
%$\cnum{2}~B$ is defined by an expression
in which all applications nest to the left, never to the right.
We call such terms {\em flat\/}~\cite{Okasaki03jfp}.
We write $\sapp{X}{k}$
for the flat term defined by $\underbrace{X~X~X~\dots~ X}_k
= \underbrace{(\dots((X~X)~X)\dots)~ X}_k$
(that can be written as \(I\,X^{\ctilde k}\) in Barendregt's notation~\cite{Barendregt84book}).
%as Statman's subscript.
%\begin{align*}
%\sapp{X}{1}&= X &
%\sapp{X}{k+1}&=\sapp{X}{k}~X \quad(k\geq 1)\text.
%\end{align*}
Using this notation, we can write $B^2 = \sapp{B}{3}$.

Okasaki~\cite{Okasaki03jfp} investigated facts about flatness.
For example, he shows that there is no universal combinator $X$ that can represent
any combinator by $\sapp{X}{k}$ with some $k$.
We shall delve into the case of $X=B$.
Consider the $n$-argument composition operator $B^n$.
%Let us consider the $3$-argument composition expressed by $B^3$.
We have already seen that $B^2$ is $\beta\eta$-equivalent to the flat term $\sapp{B}{3}$.
For $n=3$,
%we can also check $B^3 =\allowbreak B~B~B~B~B~B~B~B = \sapp{B}{8}$
%by repeating $\beta$-reduction
%for $\sapp{B}{8}~f~g~x~y~z = f~(g~x~y~z)$.
%Does $k$ exist such that $B^3 = \sapp{B}{k}$?
%The answer is yes.
using $\UL{f}~(\WL{g}~x) = B~f~g~x$, we have
\begin{align*}
B^3
&= B~B^2~B \\
&= \UL{B}~(\WL{B~B}~B)~B \\
&= \UL{B~B}~(\WL{B}~B)~B~B \\
&= \UL{B}~(\WL{B}~B)~B~B~B~B \\
&= B~B~B~B~B~B~B~B\text,
\end{align*}
and thus $B^3=\sapp{B}{8}$.
%$\cnum{3}~B=\sapp{B}{8}$.
%
%\comment{We should first say that the $n$-argument composition with $n>3$ cannot be expressed by flat $B$-terms.}
How about the $4$-argument composition $B^4$?
In fact, there is no integer $k$ such that
$B^4 = \sapp{B}{k}$
with respect to $\beta\eta$-equality.
Moreover, for any $n > 3$, there does not exist $k$
such that $B^n = \sapp{B}{k}$.
This surprising fact is proved by a quite simple method;
listing all $\sapp{B}{k}$s for $k=1,2,\dots$
and checking that none of them is equivalent to $B^n$.
An easy computation gives
$\sapp{B}{6}=\sapp{B}{10}=\Lam{x}{y}{z}{w}{v}.~x~(y~z)~(w~v)$,
and hence $\sapp{B}{i}=\sapp{B}{i+4}$ for every $i\geq6$.
%as reported in~\cite{Nakano98rims}.
%and $\sapp{B}{5}\neq\sapp{B}{9}$ will be shown later.
%
Then, by computing $\sapp{B}{k}$s only for $k\in\{1,2,\dots,6\}$,
we can check that $\sapp{B}{k}$ is not $\beta\eta$-equivalent to $B^n$ with $n>3$ for $k\in\{1,2,\dots\}$.
%$\cnum{4}~B$,
Thus we conclude that there is no integer $k$ such that $B^n =\sapp{B}{k}$.
%$\cnum{4}~B=\sapp{B}{k}$.

This is the starting point of our research.
We say that a combinator $X$ has the {\em $\rho$-property\/}
if there exist two distinct integers $i$ and $j$ such that
$\sapp{X}{i} = \sapp{X}{j}$.
In this case, we have $\sapp{X}{i+k}=\sapp{X}{j+k}$
for any $k\geq 0$
~({\`a}~la \emph{finite monogenic semigroup}~\cite{Ljapin68book}).
\reffig{rhoB} shows a computation graph of $\sapp{B}{k}$.
The $\rho$-property is named 
after the shape of the graph.
%ちなみに，別の結合子$D=B~B$も図のような$\rho$-propertyを持っている．
\begin{figure*}[t]\centering
\newcount\picW \newcount\picH \picW=420 \picH=60
\begin{picture}(\picW,\picH)(0,0)
\divide\picW 10 \advance\picW 1 \divide\picH 10 \advance\picH 1 
% \multiput(0,0)(0,10){\picH}{\multiput(0,0)(10,0){\picW}{\circle*{1}}} % grid
\def\Line(#1,#2)(#3,#4){
\newcount\middleX \middleX=#1 \advance\middleX #3 \divide\middleX 2
\newcount\middleY \middleY=#2 \advance\middleY #4 \divide\middleY 2
\qbezier(#1,#2)(\middleX,\middleY)(#3,#4)
}
\def\startX{0} \def\upperY{50}
\def\marginB{10} \def\spanBs{10} \def\spanBh{120} \def\spanBv{20}
\newcount\tmpX \newcount\tmpY
\def\putB[#1]{\put(\tmpX,\tmpY){\makebox(20,0){$\sapp{B}{#1}$}} \advance\tmpX 20}
\def\putBs[#1]{\newcount\tmpXs \tmpXs=\tmpX \newcount\tmpYs \tmpYs=\tmpY
\advance\tmpXs-2 \advance\tmpYs-12
\put(\tmpXs,\tmpYs){\makebox(100,0){$(#1)$}}}
\def\lineB(#1){
\newcount\nextX \nextX=\tmpX \advance\nextX #1
\Line(\tmpX,\tmpY)(\nextX,\tmpY) \tmpX=\nextX
}
\def\lineBs{\lineB(\spanBs)} \def\lineBh{\lineB(\spanBh)}
\def\lineBv{
\newcount\nextY \nextY=\tmpY \advance\nextY -\spanBv
\Line(\tmpX,\tmpY)(\tmpX,\nextY) \tmpY=\nextY
}
\def\turnR{\advance\tmpX-12 \advance\tmpY-8}
\def\turnL{\advance\tmpX-8 \advance\tmpY-8}
\tmpX=\startX \tmpY=\upperY
\putB[1] \lineBs \putB[2] \lineBs \putB[3] \lineBs \putB[4] \lineBs \putB[5] \lineBs
\putB[6] \putBs[=\sapp{B}{10}=\sapp{B}{14}=\dots]
\newcount\sixX \sixX=\tmpX \newcount\sixY \sixY=\tmpY \lineBh
\putB[7] \putBs[=\sapp{B}{11}=\sapp{B}{15}=\dots] \turnR \lineBv
\tmpX=\sixX \tmpY=\sixY \turnR \lineBv
\turnL \putB[9] \putBs[=\sapp{B}{13}=\sapp{B}{17}=\dots]
\lineBh \putB[8] \putBs[=\sapp{B}{12}=\sapp{B}{16}=\dots]
\end{picture}
\caption{$\rho$-property of the $B$ combinator}
\labfig{rhoB}
\end{figure*}

This paper discusses the $\rho$-property of combinatory terms,
particularly terms built from $B$ alone.
We call such terms {\em $B$-terms} and $\CL{B}$ denotes the set of all $B$-terms.
For example,
the $B$-term $B~ B$ enjoys the $\rho$-property 
with $\sapp{(B~ B)}{52}=\sapp{(B~ B)}{32}$
and so does $B~ (B~ B)$ with
$\sapp{(B~ (B~ B))}{294}=\sapp{(B~ (B~ B))}{258}$
as reported in~\cite{Nakano08trs}.
Several combinators other than $B$-terms
can be found to enjoy the $\rho$-property,
for example,
$K=\Lam{x}{y}.~ x$ and 
$C=\Lam{x}{y}{z}.~ x~ z~ y$ 
because of $\sapp{K}{3}=\sapp{K}{1}$ and 
$\sapp{C}{4}=\sapp{C}{3}$.
They are less interesting
in the sense that the cycle starts immediately and its size is very small,
comparing with $B$-terms like $B~ B$ and $B~ (B~ B)$.
As we will see later,
$B~ (B~ (B~ (B~ (B~ (B~ B))))) (\equiv B^6~ B)$ has the $\rho$-property
with the cycle of size more than $3\times 10^{11}$
which starts after more than $2\times 10^{12}$ repetitive right applications.
% 339,020,201,163
% We will see 
This is why the $\rho$-property of $B$-terms is intensively discussed
in the present paper.
A general definition of the $\rho$-property is presented in~\refsec{prelim}.
%
%From these observations, Nakano conjectured~\cite{Nakano08trs} that
%%\begin{quote}\normalsize
%``$B$-term $X$ has the $\rho$-property
%if and only if $X$ is equivalent to $B^n B$ with some $n$''.

%\comment{todo: make the following summary consistent}
The contributions of the paper are two-fold.
One is to give a characterization of $\CL{B}$ (\refsec{checkeq}) and 
another is to provide a sufficient condition for 
the $\rho$-property and anti-$\rho$-property of $B$-terms
(\refsec{results}). 
%\comment{The following sentences must be corrected!}
In the former, we introduce a canonical representation of $B$-terms
and establish a sound and complete equational axiomatization for $\CL{B}$.
In the latter, 
the $\rho$-property of $B^n B$ with $n\leq 6$ is shown with an efficient algorithm
and 
the anti-$\rho$-property for $B$-terms of particular forms
%$(B^k B)^{(k+2)n}$ with $n\geq1$ and $k\geq0$
is proved.

This paper extends and refines our paper presented in FSCD 2018~\cite{Ikebuchi18fscd}.
Compared to our previous work, we have made several improvements.
First, we add relationships to the existing work,
the Curry's compositive normal form and the Thompson's group.
Second, we report progress on proving and disproving the $\rho$-property of $B$-terms.
For proving the $\rho$-property,
we add more precise information on the implementation of our $\rho$-property checker.
For disproving the $\rho$-property,
we introduce another proof method for a specific $B$-term
and expand the set of $B$-terms which are known not to have the $\rho$-property.
Furthermore,
we discuss other possible approaches for further steps 
to show a conjecture by the second author~\cite{Nakano08trs}.

%After the definition of the $\rho$-property in \refsec{prelim},
%the following two contributions of this paper are presented.
%\begin{itemize}
%\item
%Axiomatization of $\CL{B}$: ...
%\item
%$\rho$-property of $B$-terms: ...
%\end{itemize}
%...

%
%In this paper, we prove the only-if part of this conjecture,
%while the if part (for $n\geq 7$) is left open.
%The proof is 
%based on a new result
%of a canonical representation of $B$-terms
%and several equation laws between $B$-terms.
%
%However, it is generally hard to decide by a simple computation
%whether a given combinator has the $\rho$-property or not.
%Furthermore, we will discuss when $B$-terms do not have the $\rho$-property.
%through their interesting properties.
%Let's play with the $B$ combinator!
%The contribution of the paper

%We will use the equation symbol `$=$' for $\beta\eta$-equality of combinators.

%%% Local Variables:
%%% mode: latex
%%% TeX-master: "main.tex"
%%% End:

% Rho-property of terms
%!TEX root = main.tex
\section{The $\rho$-property of terms}
\labsec{prelim}
%\comment{todo: excuse for considering the other terms here}
The $\rho$-property of a combinator $X$ is that
$\sapp{X}{i}=\sapp{X}{i+j}$ holds for some $i,j\geq1$.
%The name of the property comes from the shape of the sequence
%$\{\sapp{X}{n}\}_{n=0}^{\infty}$ as shown in \reffig{rhoB}.
We adopt $\beta\eta$-equality of corresponding $\lambda$-terms
for the equality of combinatory terms in this paper.
We could use another equality, for example,
induced by the axioms of combinatory logic.
The choice of equality is not essential here,
e.g., $\sapp{B}{9}$ and $\sapp{B}{13}$ are equal even up to the combinatory axiom of $B$,
as well as $\beta\eta$-equality.
(See \refsec{possapp} for more details.)
Furthermore, for simplicity,
we only deal with the case where $\sapp{X}{n}$
is normalizable for all $n$.
If $\sapp{X}{n}$ is not normalizable, 
it is much more difficult to check equivalence with the other terms.
This restriction does not affect the results of the paper
because all $B$-terms are normalizing.

Let us write $\rho(X)=(i,j)$ 
if a combinator $X$ has the $\rho$-property  
due to $\sapp{X}{i}=\sapp{X}{i+j}$
with minimum positive integers $i$ and $j$.
For example, we have $\rho(B)=(6,4)$,
$\rho(C)=(3,1)$, $\rho(K)=(1,2)$ and $\rho(I)=(1,1)$.
Besides them,
several combinators introduced in Smullyan's book~\cite{Smullyan12book}
have the $\rho$-property:
\begin{align*}
\rho(D)&=(32,20) 
&&\text{where $D=\Lam{x}{y}{z}{w}.\, x~ y~ (z~ w)$}\\
\rho(F)&=(3,1) &&\text{where $F=\Lam{x}{y}{z}.\, z~ y~ x$}\\
\rho(R)&=(3,1) &&\text{where $R=\Lam{x}{y}{z}.\, y~ z~ x$}\\
\rho(T)&=(2,1) &&\text{where $T=\Lam{x}{y}.\, y~ x$}\\
\rho(V)&=(3,1) &&\text{where $V=\Lam{x}{y}{z}.\, z~ x~ y$}\text.
\end{align*}
Except for the $B$ and $D~(=B~ B)$ combinators,
the property is `trivial' in the sense that
the loop starts early and the size of the cycle is very small.

%The combinator $S$ and $O$ in the book do not have the $\rho$-property.
On the other hand,
the combinators $S=\Lam{x}{y}{z}.\, x~z~(y~z)$
%$G=\lambda x.\lambda y.\lambda z.\lambda w. x~w~(y~z)$,
%$J=\lambda x.\lambda y.\lambda z.\lambda w. x~y~(x~z~w)$,
and $O=\Lam{x}{y}.\, y~(x~ y)$ in the book
do not have the $\rho$-property 
since their right application expands the \(\lambda\)-terms as illustrated by
%We can check it by the growth of the corresponding $\lambda$-terms.
%For example, 
%we can show anti-$\rho$-property of $S$ by
%which is illustrated by
\begin{align*}
\sapp{S}{2n+1} &=
\Lam{x}{y}.\, \underbrace{x~ y~ (x~ y~ (\dots (x~ y}_n~ 
(\lambda z. x~ z~ (y~ z)))\dots))
\text,
%\text{, and}
\\
\sapp{O}{n+1} &=
\lambda x. \underbrace{x~ (x~ (\dots (x~ }_n~ 
(\lambda y. y~ (x~ y))
\text.
\end{align*}
%%requires at least $n+3$ number of variables.
%
%The combinators $L=\lambda x.\lambda y. x~ (y~ y)$,
%$M=\lambda x. x~ x$, $U=\lambda x.\lambda y. y~(x~x~y)$ and
%$W=\lambda x.\lambda y. x~y~y$ do not have the $\rho$-property
%for reason~(B).
%We can check them by that 
%$\sapp{L}{4}$, $\sapp{M}{2}$, $\sapp{U}{2}$ and $\sapp{W}{4}$
%are not normalizable.

%
The definition of the $\rho$-property is naturally extended
from single combinators to terms obtained by combining several combinators.
We found by computation that
several $B$-terms
have the $\rho$-property
as shown below.
%as shown in~\reffig{rhobsequence}.
%
%\begin{figure*}
\begin{align*}
\rho(B^0 B)&=(6,4)  % 10
&\rho(B^4 B)&=(191206,431453)\\ % 622659
\rho(B^1 B)&=(32,20) % 52
&\rho(B^5 B)&=(766241307,234444571)\\ % 1000685878 (time bfast -fx -m 5 @ nack1; 1212.596u 0.003s 20:12.82 99.9%)
\rho(B^2 B)&=(258,36) % 294
&\rho(B^6 B)&=(2641033883877,339020201163)\\ % 2980054085040
\rho(B^3 B)&=(4240,5796) % 10036
\end{align*}
%\caption{$\rho$-property of $B$-terms in a particular form}
%\labfig{rhobsequence}
%\end{figure*}
%
%The implementation is based on two techniques:
%one is the Floyd's cycle-finding algorithm%
%~\cite{Knuth69book}
%(also called the tortoise and the hare algorithm)
%to detect the $\rho$-property with a constant memory usage;
%another is a canonical representation of $B$-terms,
%which will be introduced in the next section,
%to perform right application and equivalence check efficiently.
%Even by such efficient algorithms,
%it took around 59 days to find the $\rho$-property of $B^6 B$.
The details will be shown in~\refsec{rhob6b}.
%The author conjectured that
%and $B^n~ B$ has the $\rho$-property
%for any non-negative integer $n$~\cite{Nakano08trs}.

%In this paper, we investigate not only the \(\rho\)-property of \(B\)-terms 
%but also a canonical representation of the \(B\)-terms
%which plays important roles to prove or disprove the \(\rho\)-property of given \(B\)-terms.
%
From his observation on repetitive right applications 
for several \(B\)-terms,
Nakano~\cite{Nakano08trs} has conjectured as follows.
\begin{conj}
\labcnj{B-conj}
A $B$-term $e$ has the $\rho$-property if and only if 
$e$ is a monomial, \ie, $e$ is equivalent to $B^n B$ with $n\geq0$.
%That is, $\rho$-property of $B$-terms is decidable.
\end{conj}
\noindent
The definition of \emph{monomial} will be given in~\refsec{canonical}
in the context of a canonical representation of \(B\)-terms.
The ``if'' part of the conjecture for \(n\leq 6\) is shown by the above results;
the ``only if'' part will be shown for specific \(B\) terms
which will be discussed~\refsec{antib2n}. 
%(B^k B)^{(k+2)n}~\allowbreak(k\ge0,n>0)$ and $(B^2 B)^2\circ(B B)^2\circ B^2$
%has been shown by \refthr{antibgeneral}.
%This conjecture implies that the $\rho$-property of $B$-terms is decidable.
Note that the \(\rho\)-property of \(X\) 
can be rephrased in terms of the set generated by right application,
that is,
the finiteness of the set \(\{ \mathit{nf}(\sapp{X}{n}) \mid n\ge1 \}\)
where \(\mathit{nf}(e)\) represents the normal form of \(e\).
\refcnj{B-conj} claims that
for any \(B\)-term \(e\), 
the finiteness of the set \(\{\mathit{nf}(\sapp{e}{n})\mid n\geq1\}\) is decidable
since so is the word problem of \(B\)-terms.
%is decidable \(e\) is equivalent to \(B^n B\) with some \(n\geq0\).
%This statement may be helpful to consider its proof for both ``if'' and ``only-if'' part.

%%% Local Variables:
%%% mode: latex
%%% TeX-master: "main.tex"
%%% End:

% Characterization of B-terms
%!TEX root = main.tex
\section{Checking equivalence of $B$-terms}
\labsec{checkeq}
%We show that, for any $n$,
%$B$-terms equivalent to the $n$-argument composition $B^n$
%do not enjoy the $\rho$-property.
The set of all $B$-terms, $\CL{B}$,
is closed under application by definition,
that is, the repetitive right application of a $B$-term
always generates a sequence of $B$-terms.
Hence, the $\rho$-property can be decided
by checking `equivalence' among generated $B$-terms,
where the equivalence should be checked
through $\beta\eta$-equivalence of their corresponding $\lambda$-terms
in accordance with the definition of the $\rho$-property.
It would be useful if we have a fast algorithm
for deciding equivalence over $B$-terms.
%It will make it hard
%to analyze how the sequence of generated terms grows or circulates.

In this section,
we give a characterization of the $B$-terms
to efficiently decide their equivalence.
We introduce a method for deciding equivalence of $B$-terms
without calculating the corresponding $\lambda$-terms.
To this end, 
we first investigate equivalence over $B$-terms
with examples
and then present an equation system as a characterization of $B$-terms
so as to decide equivalence between two $B$-terms.
Based on the equation system, 
we introduce a canonical representation of $B$-terms.
The representation makes it easy to observe 
the growth caused by repetitive right application of $B$-terms,
which will be later used for proving 
the anti-$\rho$-property of $B^{2}$.
We believe that this representation will be helpful to prove 
the $\rho$-property or the anti-$\rho$-property 
for the other $B$-terms.

\subsection{Equivalence over $B$-terms}
Two $B$-terms are said to be \emph{equivalent} if their corresponding $\lambda$-terms
are $\beta\eta$-equivalent.
For instance, $B~ B~ (B~ B)$ and $B~ (B~ B)~ B~ B$ are equivalent.
This can be shown
% without considering their $\lambda$-terms
by the definition $B~x~y~z = x~ (y~ z)$.
For another (non-trivial) instance,
$B~ B~ (B~ B)$ and
$B~ (B~ (B~ B))~ B$ are 
equivalent.
This is illustrated
by the fact that
they are equivalent to 
$\Lam{x}{y}{z}{w}{v}.\, x~(y~z)~(w~v)$
%through the following long calculation:
%\comment{$\eta$ is not required if we fully expand all combinators as lambda terms}
%\begin{eqnarray*}
%B~ B~ (B~ B)
%&= (\lambda x.\lambda y.\lambda z. x~(y~z)) B~ B~ (B~ B)
%\\&= \lambda z. B~(B~ B~ z)
%%\\&= \lambda z. (\Lam{x}{y}. \lambda w. x~(y~w)) (B~B~z)
%\\&= \lambda z. \lambda y. \lambda w. B~B~z~(y~w)
%\\&= \lambda z. \lambda y. \lambda w. B~(z~(y~w))
%%\\&= \lambda z. \lambda y. \lambda w. 
%%     (\lambda x.\lambda v.\lambda u. x~(v~u)) (B~z~y~w)
%\\&= \lambda z. \lambda y. \lambda w. 
%     \lambda v.\lambda u. z~(y~w)~(v~u)
%\\&=
%\\&= \lambda z. \lambda w. (\lambda x.\lambda y.\lambda v. x~(y~v)) (B~z~y~w)
%\\&= \lambda z. \lambda w. B~B~(B~z~y)~w
%\end{eqnarray*}
%\begin{eqnarray*}
%B~ B~ (B~ B)
%&=_\eta \lambda x. B~ B~ (B~ B)~ x
%\\&
%=_\beta \lambda x. B~ (B~ B~ x)
%\\&
%=_\eta  \Lam{x}{y}. B~ (B~ B~ x)~ y
%\\&
%=_\eta  \Lam{x}{y}{z}. B~ (B~ B~ x)~ y~ z
%\\&
%=_\beta \Lam{x}{y}{z}. B~ B~ x~ (y~ z)
%\\&
%=_\beta \Lam{x}{y}{z}. B~ (x~ (y~ z))
%\\&
%=_\beta \Lam{x}{y}{z}. B~ (B~ x~ y~ z)
%\\&
%=_\beta \Lam{x}{y}{z}. B~ B~ (B~ x~ y)~ z
%\\&
%=_\eta  \Lam{x}{y}. B~ B~ (B~ x~ y)
%\\&
%=_\beta \Lam{x}{y}. B~ (B~ B)~ (B~ x)~ y
%\\&
%=_\eta  \lambda x. B~ (B~ B)~ (B~ x)
%\\&
%=_\beta \lambda x. B~ (B~ (B~ B))~ B~ x
%\\&
%=_\eta  B~ (B~ (B~ B))~ B\mbox.
%\end{eqnarray*}
%This example may be explained by the associativity of function composition.
%BB(BB)=BB(BB)x=B(BBx)yz=BBx(yz)=B(x(yz))
%B(B(BB))B=B(B(BB))Bx=B(BB)(Bx)=B(BB)(Bx)y=BB(Bxy)=BB(Bxy)z=B(Bxyz)
%
where $B$ is replaced with $\Lam{x}{y}{z}.\, x~(y~z)$
or the other way around
at the $=_\beta$ equation.
Similarly,
it is hard to directly show equivalence between the two $B$-terms,
\(B~(B~B)~(B~B)\) and \(B~(B~B~B)\),
which requires long calculation like:
\begin{align*}
B~ B~ (B~ B)
&=_\eta \lambda x. B~ B~ (B~ B)~ x
\\&
=_\beta \lambda x. B~ (B~ B~ x)
\\&
=_\eta  \Lam{x}{y}. B~ (B~ B~ x)~ y
\\&
=_\eta  \Lam{x}{y}{z}. B~ (B~ B~ x)~ y~ z
\\&
=_\beta \Lam{x}{y}{z}. B~ B~ x~ (y~ z)
\\&
=_\beta \Lam{x}{y}{z}. B~ (x~ (y~ z))
\\&
=_\beta \Lam{x}{y}{z}. B~ (B~ x~ y~ z)
\\&
=_\beta \Lam{x}{y}{z}. B~ B~ (B~ x~ y)~ z
\\&
=_\eta  \Lam{x}{y}. B~ B~ (B~ x~ y)
\\&
=_\beta \Lam{x}{y}. B~ (B~ B)~ (B~ x)~ y
\\&
=_\eta  \lambda x. B~ (B~ B)~ (B~ x)
\\&
=_\beta \lambda x. B~ (B~ (B~ B))~ B~ x
\\&
=_\eta  B~ (B~ (B~ B))~ B\mbox.
\end{align*}
This kind of equality makes it hard
to investigate the $\rho$-property of $B$-terms.
To solve this annoying issue,
we will introduce a canonical representation of $B$-terms in \refsec{canonical}.

\subsection{Equational axiomatization for $B$-terms}
Equality between two $B$-terms can be decided through
their canonical representation introduced in \refsec{canonical}.
The representation is based on
a sound and complete equation system
as described in
the next theorem.

\begin{thm}
\labthr{eqB}
Two $B$-terms are $\beta\eta$-equivalent
if and only if
%their equality is derived from equations~(B1), (B2), and (B3).
their equality is derived from the following equations:
%\begin{figure}
\begin{align}
B~ x~ y~ z &= x~ (y~ z)
\tag{B1}
%\labeqn{Bdelta}
\\
B~ (B~ x~ y) &= B~ (B~ x)~ (B~ y)
\tag{B2}
%\labeqn{Bdistr}
\\
B~ B~ (B~ x) &= B~ (B~ (B~ x))~ B
\tag{B3}
%\labeqn{Blift}
\end{align}
%\caption{Equational axiomatization for $B$-terms}
%\labfig{eqB}
%\end{figure}
%
%\refeqn{bdelta}, \refeqn{bdistr}, and \refeqn{blift}.
\end{thm}

The proof of the ``if'' part,
which corresponds to the soundness of the equation system (B1), (B2), and (B3),
is given here.
We will later prove the ``only if'' part
with the uniqueness of the canonical representation of $B$-terms.
\\
\begin{proof}%[\proofname~of if-part of \refthr{eqB}]
Equation~(B1) is 
immediate from the definition of $B$.
Equations~(B2) and~(B3) are shown by
\begin{align*}
B~ (B~ e_1~ e_2) &=
\Lam{x}{y}.\, B~ (B~ e_1~ e_2)~ x~ y
                                    & B~ B~ (B~ e_1) 
                                    &= \lambda x.\, B~ B~ (B~ e_1)~ x
\\&=
\Lam{x}{y}.\, B~ e_1~ e_2~ (x~ y)
                                    &&= \lambda x.\, B~ (B~ e_1~ x)
\\&=
\Lam{x}{y}.\, e_1~ (e_2~ (x~ y))
                                    &&= \Lam{x}{y}{z}.\, B~ e_1~ x (y~ z)
\\&=
\Lam{x}{y}.\, e_1~ (B~ e_2~ x~ y)
                                    &&= \Lam{x}{y}{z}.\, e_1~ (x~ (y~ z))
\\&=
\lambda x.\, B~ e_1~ (B~ e_2~ x)
                                    &&= \Lam{x}{y}{z}.\, e_1~ (B~ x~ y~ z)
\\&=
B~ (B~ e_1)~ (B~ e_2)
                                    &&= \Lam{x}{y}.\, B~ e_1~ (B~ x~ y)
\\&
                                    &&= \lambda x.\, B~ (B~ e_1)~ (B~ x)
\\&                                    
                                    &&= B~ (B~ (B~ e_1))~ B
\end{align*}
where the $\alpha$-renaming is implicitly used.
\end{proof}

Equation~(B2) has been employed by Statman~\cite{Statman10type}
to show that no $B\omega$-term can be a fixed-point combinator
where $\omega=\lambda x.\, x~x$.
This equation exposes an interesting feature of the $B$ combinator.
Write equation~(B2) as
\begin{align}
B~ (e_1 \circ e_2) = (B~ e_1) \circ (B~ e_2)
\tag{B2'}
\labeqn{Bo-distr}
\end{align}
by replacing every $B$ combinator with $\circ$ infix operator
if it has exactly two arguments.
The equation is a distributive law of $B$ over $\circ$,
which will be used to obtain the canonical representation of $B$-terms.
Equation~(B3) is also used for the same purpose
as the form of
\begin{align}
B \circ (B~ e_1) = (B~ (B~ e_1)) \circ B\text.
\tag{B3'}
\labeqn{Bo-push}
\end{align}

%Since the $B$ combinator denotes function composition,
We also have a natural equation
%\begin{equation}
$B~ e_1~ (B~ e_2~ e_3) = B~ (B~ e_1~ e_2)~ e_3$
%\labeqn{B-assoc}
%\end{equation}
which represents associativity of function composition,
i.e.,~$e_1\circ (e_2\circ e_3) = (e_1\circ e_2)\circ e_3$.
This is shown with equations~(B1) and (B2)
by
\begin{equation*}
B~ e_1~ (B~ e_2~ e_3) = B~ (B~ e_1)~ (B~ e_2)~ e_3 = B~ (B~ e_1~ e_2)~ e_3
.
\end{equation*}

\subsection{Canonical representation of $B$-terms}
\labsec{canonical}
To decide equality between two $B$-terms,
it does not suffice to compute their normal forms
under the definition of $B$, $B~x~y~z\to x~(y~z)$.
This is because
two distinct normal forms may be equal up to $\beta\eta$-equivalence,
e.g., $B~B~(B~B)$ and $B~(B~(B~B))~B$.
We introduce a canonical representation of $B$-terms,
which makes it easy to check equivalence of $B$-terms.
We will eventually find that for any $B$-term $e$
there exists a unique finite non-empty weakly-decreasing sequence
of non-negative integers
$n_1\geq n_2\geq\dots\geq n_k$ 
such that $e$ is 
equivalent
%$\beta\eta$-equivalent
to $(B^{n_1} B) \circ (B^{n_2} B)\circ \dots \circ (B^{n_k} B)$.
Ignoring the inequality condition gives 
\emph{polynomials} introduced by Statman~\cite{Statman10type}.
We will use these \emph{decreasing polynomials} for our canonical representation
as presented later. 
A similar result is found in~\cite{Curry30ajm2} as discussed later.

First, we explain how this canonical form is obtained from a $B$-term.
We only need to consider $B$-terms 
in which every $B$ has at most two arguments.
One can reduce the arguments of $B$ to less than three
by repeatedly rewriting occurrences of $B~e_1~e_2~e_3~e_4~\dots~e_n$ into
$e_1~(e_2~e_3)~e_4~\dots~e_n$.
The rewriting procedure always terminates because it reduces the number of $B$.
Thus, every $B$-term in $\CL{B}$ is equivalent to a $B$-term built
by the syntax
\begin{align}
e &::= B ~\mid~ B ~ e ~\mid~ e ~ \circ ~ e
\labeqn{syntaxBo}
\end{align}
where $e_1\circ e_2$ denotes $B~e_1~e_2$.
%Recall that the $B$ combinator plays a role of a function composition.
We prefer to use the infix operator $\circ$ instead of $B$ that has two arguments
because associativity of $B$, that is, 
$B~e_1~(B~e_2~e_3) = B~(B~e_1~e_2)~e_3$
can be implicitly assumed.
This simplifies the further discussion on $B$-terms.
We will deal with only $B$-terms in syntax~\refeqn{syntaxBo} from now on.
The $\circ$ operator has lower precedence than application in this paper,
e.g., terms $B~B\circ B$ and $B\circ B~B$ represent $(B~B)\circ B$ and $B\circ(B~B)$, respectively.

The syntactic restriction by \refeqn{syntaxBo} does not suffice
to proffer a canonical representation of $B$-terms.
%There are many pairs of $B$-terms which are equivalent
%even in the form of \refeqn{syntaxBo}.
For example, both of the two $B$-terms $B\circ B~B$ and $B~(B~B) \circ B$ are
given in the form of \refeqn{syntaxBo},
but we can see that they are equivalent using \refeqn{Bo-push}.

A \emph{polynomial form of $B$-terms} is obtained
by putting a restriction on the syntax
so that no $B$ combinator occurs outside of the $\circ$ operator
while syntax~\refeqn{syntaxBo} allows
the $B$ combinators and the $\circ$ operators to occur
in an arbitrary position.
The restricted syntax is given as
\begin{align*}
e ::= e_B ~\mid~ e \circ e 
\qquad 
e_B ::= B ~\mid B~ e_B
\end{align*}
where terms in $e_B$ have a form of
%$\underbrace{B (\dots(B}_{n} B)\dots)$, that is $B^n B$,
$B(\dots(B(B~ B))\dots)$, that is $B^n B$ with some $n$,
called \emph{monomial}.
The syntax can be simply rewritten into
$e \dcolonequals B^n B ~\mid~ e \circ e$,
which is called \emph{polynomial}.

\begin{defi}
A $B$-term $B^n B$ is called \emph{monomial}.
A \emph{polynomial} 
%of $B$-terms
is a $B$-term given in the form of
\begin{align*}
%e_i &::= B^n~ B ~\mid~ B^n~ e_n  \quad\text{where $n>i$}
(B^{n_1} B) \circ (B^{n_2} B) \circ\dots\circ (B^{n_k} B)
%B^{n_1} B \circ B^{n_2} B \circ\dots\circ B^{n_k} B
%\quad\text{where $n>i$}
\end{align*}
where $k>0$ and $n_1,\dots, n_k\geq0$ are integers.
%A $B$-term in list representation is simply called \emph{list}.
In particular,
a polynomial is called \emph{decreasing}
when $n_1\geq n_2\geq\dots\geq n_k$.
The \emph{length} of a polynomial $P$
is the number of monomials in $P$,
i.e., the length of the polynomial above is $k$.
% is defined by adding 1 to the number of $\circ$ in $P$.
The numbers $n_1, n_2,\dots,n_k$ are called \emph{degrees}.
%minimum
%maximum
\end{defi}

In the rest of this subsection,
we prove that for any $B$-term $e$ there exists
a unique decreasing polynomial equivalent to $e$.
First, we show that $e$ has an equivalent polynomial.

\begin{lemC}[\cite{Statman10type}]
\lablem{H-exist}
For any $B$-term $e$,
there exists a polynomial equivalent to $e$.
\end{lemC}
\begin{proof}
%Let $e$ be a $B$-term.
We prove the statement by induction on the structure of $e$.
In the case of $e\equiv B$, the term itself is polynomial.
In the case of $e\equiv B~e_1$,
assume that $e_1$ has equivalent polynomial
$(B^{n_1} B)\circ(B^{n_2} B)\circ\dots\circ(B^{n_k} B)$.
Repeatedly applying equation~\refeqn{Bo-distr} to $B~e_1$,
we obtain a polynomial equivalent to $B~e_1$ 
as $(B^{n_1+1} B) \circ (B^{n_2+1} B) \circ\dots\circ (B^{n_k+1} B)$.
In the case of $e\equiv e_1\circ e_2$,
assume that $e_1$ and $e_2$ have equivalent polynomials
$P_1$ and $P_2$, respectively.
A polynomial equivalent to $e$ is given
by $P_1\circ P_2$.
\end{proof}

%\begin{remark}
%A $B$-term which has list representations of length 1
%is called \emph{monomial}
%while the others are called \emph{polynomial} in~\cite{Statman10type},
%where these terminologies are defined for $\CL{B,M}$
%with $M\equiv\lambda x.~x~x$, though.
%Corollary 0.1 in the paper is a statement similar to the above.
%All monomial $B$-terms are conjectured to have $\rho$-property
%as discussed later.
%\end{remark}

Next, we show that for any polynomial $P$
there exists a decreasing polynomial equivalent to $P$.
A key equation of the proof is
\begin{align}
(B^m B) \circ (B^n B) &=
(B^{n+1} B) \circ (B^m B)
\quad
\text{when~$m<n$,}
\labeqn{B-swap}
\end{align}
which is shown by
\begin{align*}
%\lefteqn{
(B^m B) \circ (B^n B)
%}
%\\
&=
B^m (B \circ(B^{n-m} B))
\\&=
B^m (B \circ (B~(B^{n-m-1} B)))
\\&=
B^m ((B (B (B^{n-m-1} B))) \circ B)
\\&
=
(B^{n+1} B) \circ (B^m B)
\end{align*}
using equations~\refeqn{Bo-distr} and~\refeqn{Bo-push}.

\begin{lem}
\lablem{H-decr}
Any polynomial $P$ has an equivalent decreasing polynomial $P'$
such that
\begin{itemize}
\item the length of $P$ and $P'$ are equal, and
\item the lowest degrees of $P$ and $P'$ are equal.
\end{itemize}
\end{lem}
\begin{proof}
We prove the statement by induction on the length of $P$.
When the length is 1, that is, $P$ is a monomial,
$P$ itself is decreasing and the statement holds.
When the length $k$ of $P$ is greater than 1,
take $P_1$ such that $P \equiv P_1 \circ (B^n B)$.
%assume that any polynomial of length $k-1$
%has an equivalent decreasing polynomial.
From the induction hypothesis,
%taking $P_1$ such that $P \equiv P_1 \circ (B^n B)$,
there exists a decreasing polynomial
$P_1'\equiv(B^{n_1} B)\circ(B^{n_2} B)\circ\dots\circ (B^{n_{k-1}} B)$
equivalent to $P_1$,
and the lowest degree of $P_1$ is $n_{k-1}$.
If $n_{k-1}\geq n$, then $P'\equiv P_1'\circ (B^n~ B)$ is decreasing
and equivalent to $P$.
Since the lowest degrees of $P$ and $P'$ are $n$,
the statement holds.
If $n_{k-1} < n$,
$P$ is equivalent to
\begin{align*}
%\lefteqn{
(B^{n_1}~ B) \circ
\dots
\circ (B^{n_{k-1}} B) \circ(B^n B)
%}\\
&=
(B^{n_1} B) \circ
\dots
\circ (B^{n+1} B) \circ (B^{n_{k-1}} B)
%\lefteqn{
%(B^{n_1}~ B) \circ
%\dots
%\circ (B^{n_{k-2}} B) \circ (B^{n_{k-1}} B) \circ(B^n B)
%}
%\\&
%=
%(B^{n_1} B) \circ
%\dots
%\circ (B^{n_{k-2}} B) \circ (B^{n+1} B) \circ (B^{n_{k-1}} B)
\end{align*}
due to equation~\refeqn{B-swap}.
Putting the last term as $P_2\circ(B^{n_{k-1}} B)$,
the length of $P_2$ is $k-1$
and the lowest degree of $P_2$ is greater than or equal to $n_{k-1}$.
From the induction hypothesis,
$P_2$ has an equivalent decreasing polynomial $P_2'$
of length $k-1$
and the lowest degree of $P_2'$ greater than or equal to $n_{k-1}$.
Thereby we obtain
a decreasing polynomial $P_2'\circ(B^{n_{k-1}} B)$
equivalent to $P$
and the statement holds.
\end{proof}

\begin{exa}
Consider a $B$-term $e=B~(B~B~B)~(B~B)~B$.
First, applying equation~(B1),
\begin{align*}
e &= B~(B~B~B)~(B~B)~(B~B)
%\\&
= B~B~B~(B~B~(B~B))
%\\&
= B~(B~(B~B~(B~B)))
\end{align*}
so that every $B$ has at most two arguments.
Then replacing each two-argument $B$ to the infix $\circ$ operator,
obtain $B~(B~(B\circ(B~B)))$.
Applying equation~\refeqn{Bo-distr}, we have
\begin{align*}
B~(B~(B\circ(B~B)))
&=
B~((B~B)\circ(B~(B~B)))
\\&
=
(B~(B~B))\circ(B~(B~(B~B)))
\\&
= (B^2 B)\circ(B^3 B)\text.
\end{align*}
Applying equation~\refeqn{B-swap}, 
we obtain the decreasing polynomial
$(B^4 B)\circ(B^2 B)$ equivalent to~$e$.
%\hfill\qef
\end{exa}
%\comment{Examples of decreasing list representation of some $B$-terms should be here}

%The proof above illustrates that
Every $B$-term
has at least one equivalent decreasing polynomial as shown so far.
To conclude this subsection,
we show the uniqueness of decreasing polynomial equivalent to any $B$-term,
that is,
every $B$-term $e$ has no two distinct decreasing polynomials equivalent to $e$.

The proof is based on the idea that
$B$-terms correspond to unlabeled binary trees.
Let $M$ be a term which is constructed from variables $x_1,\dots,x_k$ and their applications.
Then we can show that if
the $\lambda$-term $\lambda x_1.\dotsLam{x_k}.\, M$ is in $\CL{B}$,
then $M$ is obtained by putting parentheses to some positions in the sequence
$x_1~ \dots~ x_k$.
More precisely, we have the following lemma.
\begin{lem}
\lablem{bform}
Every $B$-term is $\beta\eta$-equivalent to a $\lambda$-term
of the form
$\lambda x_1.\dotsLam x_k.~M$ with some $k>2$ 
where $M$ satisfies the following two conditions:
(1) $M$ consists of only the variables $x_1,\dots,x_k$ and their applications, and
(2) for every subterm of $M$ which is in the form of $M_1~M_2$, if $M_1$ has a variable $x_i$, then $M_2$ does not have any variable $x_j$ with $j\leq i$.
\end{lem}
\begin{proof}
%By the structural induction of $B$-terms.
We prove the statement by induction.
In the case of $e\equiv B$,
$e$ is equivalent to $\Lam{x_1}{x_2}{x_3}.\, x_1(x_2 x_3)$,
the statement holds.
In the case of $e\equiv B~e_1$,
$e$ is equivalent to $\Lam{x_1}{x_2}.\, e_1~(x_1~x_2)$.
From the induction hypothesis,
$e_1$ is equivalent to $\Lam{x_1}{x_2}\dotsLam{x_k}.\, e_1'$ where $e_1'$ satisfies the conditions (1) and (2).
Then, we can see that $e_1'\left[(x_1~x_2)/x_1\right]$ also satisfies the conditions (1) and (3).
\end{proof}

%Since every $\lambda$-term in $\CL{B}$ is linear and ordered,
%that is,
%every argument of a function is used exactly once
%in the order it was applied,
%the corresponding $\lambda$-term must be in the form of
%$\lambda x_1.\lambda x_2.\dots.\lambda x_k.~ e$
%where $e$ is built by putting parentheses 
%to appropriate positions in the sequence $x_1~ x_2~ \dots~ x_k$.
From this lemma, we see that we do not need to specify variables in $M$ and
we can simply write like
$\x~\x~(\x~\x) = ~x_1~x_2~(x_3~x_4)$.
Formally speaking, every $\lambda$-term in $\CL{B}$
uniquely corresponds to a term built from $\x$ alone by the map
$(\lambda x_1\dotsLam{x_k}.\, M) \mapsto M[\x/x_1,\dots,\x/x_k]$. 
We say an unlabeled binary tree (or simply, binary tree) for
a term built from $\x$ alone since every term built from $\x$ alone
can be seen as an unlabeled binary tree.
(A term $\x$ corresponds to a leaf and $t_1~t_2$ corresponds to
the tree with left subtree $t_1$ and right subtree $t_2$.)
To specify the applications in binary trees, we write $\<t_1,t_2>$ for the application $t_1~t_2$.
%This lemma implies that 
%every $\lambda$-term in $\CL{B}$ is characterized 
%by an unlabeled binary tree.
%Whereas the body $e$ can be regarded as
%an unlabeled binary tree consisting of application and variables,
%A $\lambda$-term in $\CL{B}$ is constructed for any unlabeled binary tree
%by putting a variable to each leaf in the order of $x_1, x_2,\dots$
%and enclosing it with $k$-fold lambda abstraction
%$\lambda x_1.\lambda x_2.\dots.\lambda x_k.~[~]$
%where $k$ is a number of leaves of the binary tree.
%Let us use the notation $\star$ for a leaf
%and $\<t_1,t_2>$ for a tree
%with left subtree $t_1$ and right subtree $t_2$.
For example,
$B$-terms $B=\Lam{x}{y}{z}.\, x~(y~z)$ and
$B~B=\Lam{x}{y}{z}{w}.\, x~y~(z~w)$ are
represented by
$\<\star,\<\star,\star>>$ and
$\<\<\star,\star>,\<\star,\star>>$, respectively.

We will present an algorithm
for constructing the corresponding decreasing polynomial from a given binary tree.
First let us define a function $\nodes_i$ with integer $i$
which maps binary trees to lists of integers:
\begin{align*}
\nodes_i(\x) &= [~]
&
\nodes_i(\<t_1,t_2>) &=
\nodes_{i+\size{t_1}}(t_2) \concat \nodes_i(t_1) \concat [i]
\end{align*}
where $\dplus$ concatenates two lists
and $\size{t}$ denotes the number of leaves.
For example,
$\nodes_0(\<\<\star,\star>,\<\star,\star>>) = [2,0,0]$ and
$\nodes_1(\<\<\star,\<\star,\star>>,\<\star,\<\star,\star>>>) = [4,4,2,1,1]$.
Informally, the $\nodes_i$ function returns a list of integers
which is obtained by labeling both leaves and nodes
in the following steps.
%A \emph{labeling} procedure for a given unlabeled binary tree
%puts a label to each node and each leaf of the tree
%in the following steps.
First each leaf of a given tree is labeled by $i,i+1,i+2,\dots$
in left-to-right order.
Then each internal node of the tree is labeled by
the same label as its leftmost descendant leaf.
The $\nodes_i$ functions return a list of labels of internal nodes
in decreasing order.
\Reffig{labeling} shows three examples of labeled binary trees obtained by this labeling procedure for $i=-1$.
Let $t_j$ ($j=1,2,3$) be the unlabeled binary tree corresponding to $e_j$.
From the labeled binary trees in \Reffig{labeling}, we have $\nodes_{-1}(t_1) = [1,-1,-1]$, $\nodes_{-1}(t_2) = [3,1,1,-1,-1]$, and $\nodes_{-1}(t_3) = [5,2,2,2,0,-1,-1,-1]$.
One may notice that a binary tree \(t_3'\) 
corresponding \(\eta\)-equivalent terms of \(e_3\) is
obtained by removing the leaf 7 and its root.
From \(\nodes_{-1}(t_3')\concat[-1]=\nodes_{-1}(t_3)\),
we have \(\nodes(t_3')=\nodes(t_3)\).
It is easy to show that the \(\nodes\) function returns the same values
for \(\eta\)-equivalent \(B\)-terms.
The length of the list equals
the number of nodes, that is, 
smaller by one than the number of variables in the $\lambda$-term.
%It is easy to show the injectivity of the $\nodes_i$ function.

\def\Leaf(#1){}
\def\Node(#1;#2){edge [#1] [#2]}
\def\Edge[#1][#2]{edge [#1] [#2]}
\begin{figure}[t]
\centering
\subcaptionbox{\parbox[t]{.17\textwidth}{Binary tree \(t_1\)\\
for \(\lambda\)-term \(e_1\)\labfig{btree1}}}[.21\linewidth]{
\hspace{-25pt}
\begin{forest}
[ ,for tree={s sep=4pt, circle, draw, minimum size=1.5em, inner sep=0pt, font=\footnotesize}
  ,Nd=\(-1\!\), [ ,Nd=\(-1\!\), [\(-1\)] [\(0\)] ]
                [ ,Nd=\(1\!\),  [\(1\)]  [\(2\)] ]
]
\end{forest}
}
\subcaptionbox{\parbox[t]{.17\textwidth}{Binary tree \(t_2\)\\
for \(\lambda\)-term \(e_2\)\labfig{btree2}}}[.29\linewidth]{
\hspace{-26pt}
\begin{forest}
[ ,for tree={s sep=10pt, circle, draw, minimum size=1.5em, inner sep=0pt, font=\footnotesize}
  ,Nd=\(-1\!\), [ ,for tree={s sep=4pt}
                  ,Nd=\(-1\!\), [\(-1\)] [\(0\)] ]
                [ ,for tree={s sep=4pt}
                  ,Nd=\(1\!\),  [ ,Nd=\(1\!\!\!\;\), [\(1\)] [\(2\)] ]
                                [ ,Nd=\(3\!\), [\(3\)] [\(4\)] ] ]
]
\end{forest}
}
\subcaptionbox{Binary tree \(t_3\) for \(\lambda\)-term \(e_3\)\labfig{btree3}}[.47\linewidth]{
\hspace{-30pt}
\begin{forest}
[ ,for tree={s sep=128pt, circle, draw, minimum size=1.5em, inner sep=0pt, font=\footnotesize}
  ,Nd=\(-1\,\), [ ,for tree={s sep=14pt}
                  ,Nd=\(-1\), [ ,Nd=\(-1\!\), [\(-1\)] 
                                                [ ,for tree={s sep=4pt}
                                                  ,Nd=\(0\!\), [\(0\)] [\(1\)] ] ]
                                [ ,Nd=\(2\,\),  [ ,for tree={s sep=16pt}
                                                  ,Nd=\(2\!\), [ ,for tree={s sep=4pt}
                                                                 ,Nd=\(2\!\), [\(2\)] [\(3\)] ]
                                                               [\(4\)] ]
                                            [ ,for tree={s sep=4pt}
                                              ,Nd=\(5\!\), [\(5\)] [\(6\)] ] ] ]
              [\(7\)]
]  
\end{forest}
}
\\
\centering{
\begin{align*}
\text{where}\;\,\quad
e_1 &= \Lam{x_1}{x_2}{x_3}{x_4}.\,
x_1~x_2~(x_3~x_4)
\\
e_2 &= \Lam{x_1}{x_2}{x_3}{x_4}{x_5}{x_6}.\,
x_1~x_2~(x_3~x_4~(x_5~x_6))
\\
e_3 &= \Lam{x_1}{x_2}{x_3}{x_4}{x_5}{x_6}{x_7}{x_8}{x_9}.\,
x_1~(x_2~x_3)~(x_4~x_5~x_6~(x_7~x_8))~x_9
\\&(=_\eta \Lam{x_1}{x_2}{x_3}{x_4}{x_5}{x_6}{x_7}{x_8}.\,
x_1~(x_2~x_3)~(x_4~x_5~x_6~(x_7~x_8)))
\\\text{and}\quad
\nodes(t_1) &= [1]
\\
\nodes(t_2) &= [3,1,1]
\\
\nodes(t_3) &= [5,2,2,2,0]\\[-7ex]
\end{align*}
}
	\caption{Labeled binary trees}
	\labfig{labeling}
\end{figure}

\begin{defi}
$\nodes$ is the function which takes a binary tree $t$ and
returns the list of non-negative integers in $\nodes_{-1}(t)$,
that is, the list obtained by excluding trailing all $-1$'s in $\nodes_{-1}(t)$.
\end{defi}
%
%Note that by excluding the label $-1$'s
%it may happen to be $\nodes(t) = \nodes(t')$ for
%two distinct binary trees $t$ and $t'$
%even though the $\nodes_i$ function is injective.
%However, those binary trees $t$ and $t'$ must be `$\eta$-equivalent'
%in terms of the corresponding $\lambda$-terms.

The following lemma claims that the $\nodes$ function computes
a list of degrees of a decreasing polynomial corresponding to a given $\lambda$-term.

\begin{lem}
\lablem{H-uniq}
A decreasing polynomial $(B^{n_1} B)\circ(B^{n_2} B)\circ\dots\circ(B^{n_k} B)$
is $\beta\eta$-equivalent to the $\lambda$-term $e\in\CL{B}$
corresponding to a binary tree $t$ such that
$\nodes(t) = [n_1,n_2,\dots,n_k]$.
\end{lem}
\begin{proof}
We prove the statement by induction on the length of the polynomial $P$.

When $P\equiv B^{n} B$ with $n\geq0$,
it is found to be equivalent to the $\lambda$-term
\begin{equation*}
%e\equiv
\Lam{x_1}{x_2}{x_3}\dotsLam{x_{n+1}}\dotLam{x_{n+2}}\dotLam{x_{n+3}}.\,
x_1~
x_2~ 
x_3~ 
\dots~ x_{n+1}~ (x_{n+2}~ x_{n+3})
\end{equation*}
by induction on $n$.
This $\lambda$-term corresponds to the binary tree
$t = \<\poly{n~{\rm leaves}},\<\star,\star>>$.
Then we have that $\nodes(t) = [n]$ holds
from $\nodes_{-1}(t) = [n,\underbrace{-1,-1,\dots,-1}_{n+1}]$.

When $P\equiv P'\circ(B^{n} B)$ with
$P'\equiv(B^{n_1} B)\circ\dots\circ(B^{n_{k}} B)$,
$k\geq 1$ and $n_1\geq\dots\geq n_{k}\geq n\geq0$,
there exists a $\lambda$-term $\beta\eta$-equivalent to $P'$
corresponding to a binary tree $t'$ such that $\nodes(t')=[n_1,\dots,n_{k}]$
from the induction hypothesis.
The binary tree $t'$ must have
the form of $\<\<\poly{n_k~{\rm leaves}},t_1>,\dots,t_m>$
with $m\geq1$ and some trees $t_1,\dots,t_m$,
otherwise $\nodes(t')$ would contain an integer smaller than $n_k$.
From the definition of $\nodes$ and $\nodes_i$, we have
\begin{align}
\nodes(t') &= \nodes_{s_m}(t_m) \concat \dots \concat \nodes_{s_1}(t_1)
\labeqn{nodes_t'}
\end{align}
where $s_j = n_k + 1 + \sum_{i=1}^{j-1}\size{t_i}$.
Additionally, the structure of $t'$ implies
%\begin{eqnarray*}
$
P' =
\lambda x_1.\dotsLam{x_{l}}.\,\allowbreak
x_1~x_2~\dots~x_{n_k+1}~e_1\dots e_m
$
%\end{eqnarray*}
where $e_i$ corresponds to a binary tree $t_i$ for $i=1,\dots,m$.
From $B^{n}~B=
\lambda y_1.\dotsLam{y_{n+3}}.\,\allowbreak
y_1~ y_2\dots y_{n+1}~ (y_{n+2}~ y_{n+3})$,
we compute a $\lambda$-term $\beta\eta$-equivalent to $P\equiv P'\circ(B^n B)$ by
\begin{align*}
P
&=
\lambda x.~ P' (B^{n} B~ x)
\\&=
\lambda x.~
(\lambda x_1\dotsLam{x_{l}}.~
x_1~x_2\dots x_{n_k+1}~e_1~\dots~e_{m})
\\&\qquad\qquad
(\lambda y_2\dotsLam{y_{n+3}}.~
x~ y_2\dots y_{n+1}~ (y_{n+2}~ y_{n+3}))
\\&=
\Lam{x}{x_2}\dotsLam{x_{l}}.~
%\\&\qquad
(\lambda y_2\dotsLam{y_{n+3}}.~
x~ y_2\dots y_{n+1}~ (y_{n+2}~ y_{n+3}))
%\\&\qquad
~x_2\dots x_{n_k+1}~e_1\dots e_{m}
%\\&=
%\lambda x.
%\lambda x_2.\dots.\lambda x_{l}.~
%\\&\qquad
%(\lambda y_2.\dots.\lambda y_{n+3}.~
%x~ y_2~ \dots~ y_{n+1}~ (y_{n+2}~ y_{n+3}))
%\\&\qquad
%~x_2~\dots~x_n~\dots~x_{n_k+1}~\etree1~\dots~\etree{m}
\\&=
\Lam{x}{x_2}\dotsLam{x_{l}}.\,
\\&\qquad
(\Lam{y_{n+1}}{y_{n+2}}{y_{n+3}}.~
x~ x_2\dots x_n~y_{n+1}~(y_{n+2}~y_{n+3}))~
%\\&\qquad
x_{n+1}\dots x_{n_k+1}~e_1\dots e_{m}
\end{align*}
where $n_k\geq n$ is taken into account.
We split the proof into four cases:
(i) $n_k = n$ and $m=1$, 
(ii) $n_k = n$ and $m>1$, 
(iii) $n_k = n+1$, and
(iv) $n_k > n+1$.
In the case~(i) where $n_k = n$ and $m=1$, we have
\begin{align*}
P &= 
\Lam{x}{x_2}\dotsLam{x_{l}}\dotLam{y_{n+3}}.~
%\\&\qquad
x~ x_2\dots x_n~x_{n+1}~(e_1~y_{n+3})\text.
\end{align*}
whose corresponding binary tree $t$ is 
$\<\poly{\mbox{$n$ leaves}},\<t_1,\star>>$.
From equation~\refeqn{nodes_t'},
$\nodes(t) = \nodes_{n+1}(t_1)\dplus [n+1]
= \nodes(t')\dplus[n+1] = [n_1,\dots,n_k,n+1]$, thus the statement holds.
In the case~(ii) where $n_k = n$ and $m>1$, we have
\begin{align*}
P &= 
\Lam{x}{x_2}\dotsLam{x_{l}}.~
%\\&\qquad\qquad
x~ x_2\dots x_n~x_{n+1}~(e_1~e_2)~e_3\dots e_{m}\text.
\end{align*}
whose corresponding binary tree $t$ is 
$\<\<\poly{\mbox{$n$ leaves}},\<t_1,t_2>,t_3>,\dots,t_m>$.
Hence, $\nodes(t)=\nodes(t')\dplus[n+1]$ holds again from equation~\refeqn{nodes_t'}.
In the case~(iii) where $n_k = n+1$, we have
\begin{align*}
&P = 
\Lam{x}{x_2}\dotsLam{x_{l}}.~
x~ x_2\dots x_n~x_{n+1}~(x_{n+2}~e_1)~e_2\dots e_{m}\text{, or}
\end{align*}
whose corresponding binary tree $t$ is 
$\<\<\poly{\mbox{$n$ leaves}},\<\star,t_1>,t_2>,\dots,t_m>$.
Hence, $\nodes(t)=\nodes(t')\dplus[n+1]$ holds from equation~\refeqn{nodes_t'}.
In the case~(iv) where $n_k\geq n+2$, we have
\begin{align*}
P = 
\Lam{x}{x_2}\dotsLam{x_{l}}.~
x~ x_2\dots x_n~x_{n+1}~(x_{n+2}~x_{n+3})~\dots~e_1\dots e_{m}\text,
\end{align*}
whose corresponding binary tree $t$ is 
$\<\<\poly{\mbox{$n$ leaves}},\<\star,\star>,\dots,t_1>,\dots,t_m>$.
Hence, $\nodes(t)=\nodes(t')\dplus[n+1]$ holds from equation~\refeqn{nodes_t'}.
\end{proof}

\begin{exa}
Consider the $\lambda$-terms $e_1,e_2,e_3$ given in \Reffig{labeling}.
The $\lambda$-terms $e_1$, $e_2$, and $e_3$ given in \Reffig{labeling} are $\beta\eta$-equivalent to $B^1B$, $(B^3B)\circ(B^1B)\circ(B^1B)$, and $(B^5B)\circ(B^2B)\circ(B^2B)\circ(B^2B)\circ(B^0B)$, respectively,
since $\nodes(t_1) = [1]$, $\nodes(t_2) = [3,1,1]$, $\nodes(t_3) = [5,2,2,2,0]$. (Recall $t_j$ ($j=1,2,3$) is the unlabeled binary tree corresponding to $e_j$)
\end{exa}

We conclude the uniqueness of 
decreasing polynomials for $B$-terms shown in the following theorem.

\begin{thm}
\labthr{canonical}
Every $B$-term $e$ has a unique decreasing polynomial.
\end{thm}
\begin{proof}
For any given $B$-term $e$, 
we can find a decreasing polynomial for $e$
from \reflem{H-exist} and \reflem{H-decr}.
Since no other decreasing polynomial can be equivalent to $e$
from \reflem{H-uniq}, the present statement holds.
\end{proof}

This theorem implies that
the decreasing polynomial of $B$-terms can be used 
as their \emph{canonical representation},
which is effectively derived as shown in
\reflem{H-exist} and \reflem{H-decr}.

As a corollary of the theorem,
we can show the ``only if'' statement of \refthr{eqB},
which corresponds to the completeness of the equation system.

\begin{proof}[\proofname ~of \refthr{eqB}]
Let $e_1$ and $e_2$ be equivalent $B$-terms,
that is, their $\lambda$-terms are $\beta\eta$-equivalent.
From \refthr{canonical},
their decreasing polynomials are the same.
Since the decreasing polynomial is derived from $e_1$ and $e_2$
by equations~(B1), (B2), and (B3)
according to the proofs of \reflem{H-exist} and \reflem{H-decr},
equivalence between $e_1$ and $e_2$ is also derived from these equations.
\end{proof}

\paragraph{\textbf{Comparison with Curry's compositive normal form}}
Curry~\cite{Curry30ajm2} has introduced a similar normal form for terms 
built from regular combinators%
\footnote{
A regular combinator is a combinator in which no lambda abstraction
occurs inside function application.
}%
, including $B$-terms.
Curry's normal form is called~\emph{compositive}
~\cite{Piperno89tcs}
since it is given
as a composition of four special terms, a $K$-term, $W$-term, $C$-term, and $B$-term.
%where the $B$-term plays a role in insertion of parentheses as preprocessing.
A $B$-term in Curry's normal form is expressed by
\[
%(B^{n_1} B)^{m_1}\circ(B^{n_2} B)^{m_2}\circ\dots\circ(B^{n_k} B)^{m_k}
(B^{n_1} B^{m_1})\circ(B^{n_2} B^{m_2})\circ\dots\circ(B^{n_k} B^{m_k})
\]
where $k>0$, $n_1>n_2>\dots>n_k\geq0$ and $m_i>0$ for any $i=1,\dots,k$.
Since we have
\[
B^{n} B^{m}
= B^{n}(\underbrace{B\circ\dots\circ B}_{m})
= \underbrace{(B^n B)\circ\dots\circ(B^n B)}_{m}
\]
because of equation~(B2'),
the form is equivalent to 
%$(B^{n_1} B)^{m_1}\circ(B^{n_2} B)^{m_2}\circ\dots\circ(B^{n_k} B)^{m_k}$,
%that is,
\[
\underbrace{(B^{n_1} B)\circ\dots\circ(B^{n_1} B)}_{m_1}\circ
\underbrace{(B^{n_2} B)\circ\dots\circ(B^{n_2} B)}_{m_2}\circ\dots\circ
\underbrace{(B^{n_k} B)\circ\dots\circ(B^{n_k} B)}_{m_k}
\]
which gives a decreasing polynomial.
Curry informally proved the uniqueness of the normal form
by an observation that 
$B^{n} B^{m}=\lambda x_0.\dotsLam{x_{n+m+1}}. x_0\dots x_n~(x_{n+1}\dots x_{n+m+1})$,
while we have shown the exact correspondence
between a $B$-terms as a lambda term and
its normal form in decreasing polynomial representation.

%\comment{Revised up to here}
%!TEX root = main.tex
\subsection{Relationship with Thompson's Group}
In this subsection, we explore a relationship between polynomials and \emph{Thompson's group} $F$ \cite{McKenzie73word}.
Thompson's group $F$ is defined to be the group generated by formal elements $x_n$ ($n = 0,1,\dots$)
with  relations $x_mx_n = x_nx_{m+1}$ for any $m > n$.
Consider the map
\[
	f : \CL B \ni (B^{n_1}B)\circ \dots \circ (B^{n_k}B) \mapsto x_{n_1}^{-1}\dots x_{n_k}^{-1} \in F.
\]
The map $f$ is well-defined since for any $m > n$,
\[
f((B^nB)\circ (B^mB)) = x_n^{-1}x_m^{-1} = (x_mx_n)^{-1} = (x_nx_{m+1})^{-1} = x_{m+1}^{-1}x_n^{-1} = f((B^{m+1}B)\circ (B^nB)).
\]
We can think of $(\CL B,\circ)$ as a semigroup since $(X \circ Y) \circ Z = X \circ (Y \circ Z)$ for any $X,Y,Z\in\CL B$,
and $f : \CL B \rightarrow F$ is a semigroup homomorphism under this semigroup structure of $\CL B$. 
By definition, $f$ is a semigroup isomorphism between $\CL B$ and the subsemigroup $N$ of $F$ generated by $x_n^{-1}$ ($n=0,1,\dots$).

It is known \cite{Belk04thesis} that every element of $N$ corresponds to an infinite sequence of binary trees $(t_0,t_1,\dots)$ (called a binary forest) where there exists $k_0$ such that $t_k = \x$ for any $k \ge k_0$.
\begin{defi}
The binary forest representation of an element of $N$ is defined inductively as follows.
\begin{enumerate}
\item The binary forest representation of $x_n^{-1}$ is $(\underbrace{\x,\dots,\x}_n,\<\x,\x>,\x,\dots)$.
%\begin{figure}[hbtp]
%\begin{center}
%\includegraphics[scale=0.5]{forest_xn.pdf}
%\end{center}
%\end{figure}
\item If $y \in N$ corresponds to the binary forest $(t_0,t_1,\dots)$, $yx_n^{-1}$ corresponds to the binary forest
\[
	(t_0,t_1,\dots,t_{n-1},\<t_n,t_{n+1}>,t_{n+2},\dots).
\]
\end{enumerate}
We can see the binary forests corresponding to $x_n^{-1}x_m^{-1}$ and $x_{m+1}^{-1}x_n^{-1}$ are equal to each other for any $n,m$.
\end{defi}
(In fact, \cite{Belk04thesis} gave forest representations for the elements in the submonoid of $F$ generated by $x_n$ ($n=0,1,\dots$), not $x_n^{-1}$).
We show the binary forest representation of $x_{n_1}^{-1}\dots x_{n_k}^{-1}$ can be obtained from the binary tree corresponding to the $\lambda$-term of $(B^{n_1}B)\circ \dots \circ (B^{n_k}B)$.
\begin{thm}
Let $\<\dots\<\<\x,t_1>,t_2>\dots,t_k>$ be the binary tree corresponding to the $\lambda$-term of the polynomial $(B^{n_1}B)\circ \dots \circ (B^{n_k}B)$.
Then, the binary forest representation of $f((B^{n_1}B)\circ \dots \circ (B^{n_k}B)) = x_{n_1}^{-1}\dots x_{n_k}^{-1}$ is given by
\[
	(t_1,t_2,\dots,t_k,\x,\x,\dots).
\]
\end{thm}
\begin{proof}
We prove the theorem by induction on $k$.
For binary trees $t_1,t_2,\dots,t_m$, we write $\varphi(t_1,t_2,\dots,t_m)$ for the binary tree $\<\dots\<\<\x,t_1>,t_2>,\dots,t_m>$.
Since the binary tree corresponding to the $\lambda$-term of $B^nB$ is given by $\varphi(\underbrace{\x,\dots,\x}_n,\<\x,\x>)$,
the statement holds for the binary forest representations of $x_n = f(B^nB)$.
Suppose $n_1 \ge \dots \ge n_k \ge n_{k+1}$.
Then, the binary forest representation of $x_{n_1}^{-1}\dots x_{n_k}^{-1}x_{n_{k+1}}^{-1}$ is in the form of
$(\underbrace{\x,\dots,\x}_{n_{k+1}},\<t_1,t_2>,t_3,\dots,t_m,\x,\dots)$.
The binary tree $t = \varphi(\underbrace{\x,\dots,\x}_{n_{k+1}},\<t_1,t_2>,t_3,\dots,t_m)$ satisfies
$\mathcal{L}(t) = [n_1,\dots,n_k,n_{k+1}]$ if the binary tree
$t' = \varphi(\underbrace{\x,\dots,\x}_{n_{k+1}}, t_1, t_2,t_3,\dots,t_m)$ satisfies $\mathcal{L}(t') = [n_1,\dots,n_k]$.
By \reflem{H-uniq}, $t$ is the binary tree corresponding to the $\lambda$-term of $(B^{n_1}B)\circ\dots\circ (B^{n_{k+1}}B)$,
and this implies the desired result.
\end{proof}

%%% Local Variables:
%%% mode: latex
%%% TeX-master: "main.tex"
%%% End:

%%% Local Variables:
%%% mode: latex
%%% TeX-master: "main.tex"
%%% End:

% Results
%!TEX root = main.tex
\section{Results on the $\rho$-property of $B$-terms}
\labsec{results}
%Nakano~\cite{Nakano08trs} conjectured that
%%It has been conjectured
%``$B$-term $e$ has the $\rho$-property
%if and only if $e$ is equivalent to $B^n~B$ with some $n$''.
%In terms of decreasing polynomial representation,
%this statement can be rephrased as
%``$B$-term $e$ has the $\rho$-property
%if and only if its polynomial representation of $e$ has length 1''.
%
In this section
we show several approaches to if- and only-if-parts of 
\refcnj{B-conj}
%the conjecture
for their special cases.
%
%We investigate the $\rho$-property of concrete $B$-terms,
%some of which have the property and others do not.
For $B$-terms having the $\rho$-property,
we introduce an efficient implementation 
to compute the entry point and the size of the cycle.
For $B$-terms not having the $\rho$-property,
we give two methods for proving that they do not have it.
%we introduce an efficient algorithm based on decreasing polynomials.
%For disproving the $\rho$-property,
%
%
%It has been conjectured~\cite{Nakano08trs} that
%``$B$-term $e$ has the $\rho$-property
%if and only if $e$ is equivalent to $B^n~B$ with some $n$''.
%In terms of list representation,
%this statement can be rephrased as
%``$B$-term $e$ has the $\rho$-property
%if and only if the list representation of $e$ has length 1''.
%
%This section approaches to if- and only-if-parts of the conjecture.

\subsection{$B$-terms having the $\rho$-property}
\labsec{rhob6b}
As shown in~\refsec{prelim},
we can check that
$B$-terms equivalent to $B^n B$ with $n\leq 6$
have the $\rho$-property
by computing $\sapp{(B^n B)}{i}$ for each $i$.
However, it is not easy to check it by computer
without an efficient implementation
because we should compute all $\sapp{(B^6 B)}{i}$ with
$i\leq 2980054085040~ (= 2641033883877+339020201163)$
to know $\rho(B^6 B)=(2641033883877,339020201163)$.
A naive implementation
which computes terms of $\sapp{(B^6 B)}{i}$ for all $i$
and stores all of them
has no hope to detect the $\rho$-property.

We introduce an efficient procedure
to find the $\rho$-property of $B$-terms
which can successfully compute $\rho(B^6 B)$.
The procedure is based on two orthogonal ideas,
Floyd's cycle-finding algorithm~\cite{Knuth69book}
and
an efficient right application algorithm over decreasing polynomials
presented in~\refsec{canonical}.

The first idea, Floyd's cycle-finding algorithm
(also called the tortoise and the hare algorithm),
enables us to detect the cycle with constant memory usage,
that is, the history of all terms $\sapp{X}{i}$ does not
need to be stored to check the $\rho$-property of the $X$ combinator.
The key to this algorithm is the fact that
there are two distinct integers $i$ and $j$ with $\sapp{X}{i}=\sapp{X}{j}$
if and only if 
there is an integer $m$ with $\sapp{X}{m}=\sapp{X}{2m}$,
where the latter requires to compare $\sapp{X}{i}$ and $\sapp{X}{2i}$
from smaller $i$ and store only these two terms
for the next comparison between
$\sapp{X}{i+1}=\sapp{X}{i}X$
and $\sapp{X}{2i+2}=\sapp{X}{2i}XX$
when $\sapp{X}{i}\neq\sapp{X}{2i}$.
The following procedure computes the entry point and the size of the cycle
if $X$ has the $\rho$-property.
\begin{enumerate}
\item Find the smallest $m$ such that $\sapp{X}{m}=\sapp{X}{2m}$.
\item Find the smallest $k$ such that $\sapp{X}{k}=\sapp{X}{m+k}$.
\item Find the smallest $0<c\leq k$ such that $\sapp{X}{m}=\sapp{X}{m+c}$.
If not found, put $c=m$.
\end{enumerate}
After this procedure, we find $\rho(X)=(k,c)$.
The third step can be run in parallel during the second one.
%we find $\rho(X)=(k,\frac{m}{r})$
%where $r\geq1$ is the number of all integers 
%$0\leq i<k$ with $\sapp{X}{m}=\sapp{X}{m+i}$
%which can be counted during the second step.
See~\cite[exercise 3.1.6]{Knuth69book} for the detail.
Although we have tried the other cycle detection algorithm
developed by Brent~\cite{Brent80bit} and Gosper~\cite[item 132]{Beeler72mit},
they show a similar performance.
%it does not so 
%One could use slightly more (possibly) efficient algorithms
%for cycle detection.

Efficient cycle-finding algorithms do not suffice to compute $\rho(B^6 B)$.
Only with the idea above
running on a laptop 
%(1.7 GHz Intel Core i7 / 8GB of memory),
(2.7 GHz Intel Core i7 / 16GB of memory),
it takes about 2 hours even for $\rho(B^5 B)$ and
fails to compute $\rho(B^6 B)$ with an out-of-memory error.

The second idea
enables us to compute $\sapp{X}{i+1}$ efficiently from $\sapp{X}{i}$
for $B$-terms $X$.
The key to this algorithm is to use the canonical representation
of $\sapp{X}{i}$, that is a decreasing polynomial,
and directly compute the canonical representation of $\sapp{X}{i+1}$
from that of $\sapp{X}{i}$.
Additionally, the canonical representation enables us
to quickly decide equivalence
which is required many times to find the cycle.
It takes time just proportional to their lengths.
If the $\lambda$-terms are used for finding the cycle,
both application and deciding equivalence
require much more complicated computation.
Our implementation based on these two ideas 
computes $\rho(B^5 B)$ and $\rho(B^6 B)$
%in 10 minutes and 59 days (!), respectively.
%in 5 minutes and 26 days, respectively.  % 260 sec
in 2 minutes and 6 days, respectively.  % 119 sec
%Using this implementation,
%We can compute $\rho(B^5 B)$ within 10 minutes by this improvement.
%The $\rho$-property of $B^6 B$ 
%Even for 
%it took around 59 days to find 

%!TEX root = main.tex
%\section{Proving (anti-)$\rho$-property of $B$-terms}
%\subsection{Application on canonical representation of $B$-terms}
%The $\rho$-property of a term can be investigated
%by sharp observation on its right application.
%Now we observe application of one $B$-term to another
%in terms of canonical representation.
%It will make easy to prove or disprove
%the $\rho$-property of a given $B$-term.

%The anti-$\rho$ property of 
%circular ...
%
For two given decreasing polynomials $P_1$ and $P_2$,
we show how a decreasing polynomial $P$ equivalent to $(P_1~P_2)$
can be obtained.
The method is based on the following lemma
about an application of one $B$-term to another $B$-term.
%
%\comment{Revised up to here}
%
\begin{lem}
\lablem{comp-app}
For $B$-terms $e_1$ and $e_2$,
there exists $k\geq0$ such that
%\begin{equation*}
$e_1\circ(B~ e_2) = B~ (e_1~ e_2)\circ B^k$.
%\end{equation*}
%holds with some $k\geq0$.
\end{lem}
\begin{proof}
Let $P_1$ be a decreasing polynomial equivalent to $e_1$.
We prove the statement by case analysis on the maximum degree in $P_1$.
When the maximum degree is 0, 
we can take $k'\geq1$ such that
$P_1\equiv\underbrace{B\circ\dots\circ B}_{k'}=B^{k'}$.
Then,
\begin{align*}
e_1\circ(B\, e_2)&= 
\underbrace{B\circ\dots\circ B}_{k'}
\circ(B~ e_2) 
%\\&
= (B^{k'+1}\, e_2)\circ\underbrace{B\circ\dots\circ B}_{k'}
= B\, (e_1\, e_2)\circ B^{k'}
\end{align*}
where equation~\refeqn{Bo-push} is used $k'$ times in the second equation.
Therefore the statement holds by taking $k=k'$.
When the maximum degree is greater than 0,
we can take a decreasing polynomial $P'$ for a $B$-term and $k'\geq0$ such that
$P_1=(B\, P')\circ \underbrace{B\circ\dots\circ B}_{k'}=
(B\, P')\circ B^{k'}$ due to equation~\refeqn{Bo-distr}.
Then,
\begin{align*}
e_1\circ(B\, e_2)&= 
(B\, P')\circ\underbrace{B\circ\dots\circ B}_{k'}\circ(B\, e_2)
\\&=
(B\, P')\circ(B^{k'+1}\, e_2)\circ\underbrace{B\circ\dots\circ B}_{k'}
\\&=
B~ (P'\circ(B^{k'}\, e_2))\circ B^{k'}
\\&=
B~ (B\,P'\,(B^{k'}\, e_2))\circ B^{k'}
\\&=
B~ (P_1\, e_2)\circ B^{k'}
\\&
=
B~ (e_1\, e_2)\circ B^{k'}\text.
\end{align*}
%Since $e_1~ e_2=P_1~ e_2=B~ P'~(B^{k'}~ e_2)=P'\circ(B^{k'}~ e_2)$,
Therefore, the statement holds by taking $k=k'$.
%We have $h_1=B^{n_1}
\end{proof}

This lemma indicates that,
from two decreasing polynomials $P_1$ and $P_2$,
a decreasing polynomial $P$ equivalent to $(P_1\,P_2)$
can be obtained in the following steps
where $L_1$ and $L_2$ are lists of non-negative numbers
as shown in \refsec{canonical} corresponding to $P_1$ and $P_2$.
%respectively.
\begin{algoC}[\textbf{(Application of \(B\)-terms \(P_1\) and \(P_2\) in canonical representation)}]
\labalg{appcanon}\leavevmode
\begin{enumerate}
\item Build $P_2'$ by raising each degree of $P_2$ by 1,
\ie,~when $P_2\equiv (B^{n_1}\, B)\circ\dots\circ(B^{n_l}\, B)$,
$P_2' \equiv(B^{n_1+1}\, B)\circ\dots\circ(B^{n_l+1}\, B)$.
In terms of the list representation,
a list $L_2'$ is built from $L_2$ by incrementing each element.
\item Find a decreasing polynomial $P_{12}$
corresponding to $P_1\circ P_2'$ by equation~\refeqn{B-swap}.
In terms of the list representation,
a list $L_{12}$ is constructed
by appending $L_1$ and $L_2'$ and
repeatedly applying~\refeqn{B-swap}.
%from its tail.
%
\item Obtain $P$ by lowering each degree of $P_{12}$ after
eliminating the trailing 0-degree units,
i.e.,~when $P_{12}\equiv (B^{n_1}\, B)\circ\dots\circ(B^{n_l}\, B)
\circ(B^{0}\, B)\circ\dots\circ(B^{0}\, B)$ with $n_1\geq\dots\geq n_l>0$,
$P\equiv(B^{n_1-1}\, B)\circ\dots\circ(B^{n_l-1}\, B)$.
In terms of the list representation,
a list $L$ is obtained from $L_{12}$
by decrementing each element after removing trailing 0's.
\end{enumerate}
\end{algoC}
In the first step,
a decreasing polynomial $P_2'$ equivalent to $B\,P_2$
is obtained.
The second step yields
a decreasing polynomial $P_{12}$
for $P_1\circ P_2'=P_1\circ(B\,P_2)$.
Since $P_1$ and $P_2$ are decreasing,
it is easy to find $P_{12}$
by repetitive application of equation~\refeqn{B-swap}
for each unit of $P_2'$,
{\`a}~la insertion operation in insertion sort.
In the final step,
a polynomial $P$ that satisfies
$(B\, P)\circ B^{k} = P_{12}$ with some $k$ is obtained.
From \reflem{comp-app} and the degree of decreasing polynomials, 
$P$ is equivalent to $(P_1\,P_2)$.

\begin{exa}
Let $P_1$ and $P_2$ be decreasing polynomials represented by lists
$L_1 = [4, 1, 0]$ and $L_2 = [2, 0]$.
%$P_1 = (B^4 B)\circ(B^1 B)\circ(B^0 B)$
%and
%$P_2 = (B^2 B)\circ(B^0 B)$.
Then a decreasing polynomial $P$ equivalent to $(P_1\,P_2)$
is obtained as a list $L$ in three steps:
\begin{enumerate}
\item A list $L_2'= [3, 1]$ is obtained from $L_2$.
%$P_2'=(B^3 B)\circ(B^1 B)$.
\item
A decreasing list $L_{12}$ is obtained by
\begin{align*}
L_{12}
&= [4,1,\underline{0,3},1]
%\\&
= [4,\underline{1,4},0,1] 
%\\&
= [\underline{4,5},1,0,1]
%\\&
= [6,4,1,\underline{0,1}]
%\\&
= [6,4,\underline{1,2},0]
%\\&
= [6,4,3,1,0]
\end{align*}
%
%Find a decreasing polynomial by
%\begin{align*}
%%\lefteqn{P_1\circ P_2'}\\&=
%P_1\circ P_2'&=
%(B^4 B)\circ(B^1 B)\circ(B^0 B)\circ(B^3 B)\circ(B^1 B)
%\\&=
%(B^4 B)\circ(B^1 B)\circ(B^4 B)\circ(B^0 B)\circ(B^1 B)
%\\&=
%(B^4 B)\circ(B^5 B)\circ(B^1 B)\circ(B^0 B)\circ(B^1 B)
%\\&=
%(B^6 B)\circ(B^4 B)\circ(B^1 B)\circ(B^0 B)\circ(B^1 B)
%\\&=
%(B^6 B)\circ(B^4 B)\circ(B^1 B)\circ(B^2 B)\circ(B^0 B)
%\\&=
%(B^6 B)\circ(B^4 B)\circ(B^3 B)\circ(B^1 B)\circ(B^0 B)
%\end{align*}
where equation~\refeqn{B-swap} is applied in each underlined pair.
\item
A list $L=[5,3,2,0]$ is obtained
from $L_{12}$
as the result of the application
by decrementing each element after removing trailing 0's.
\end{enumerate}
%\hfill\qef
\end{exa}

%%% Local Variables:
%%% mode: latex
%%% TeX-master: "main.tex"
%%% End:

The implementation based on the right application over decreasing polynomials is available at
\url{https://github.com/ksk/Rho} as a program named \texttt{bpoly}.
In the current implementation,
every decreasing polynomial is represented by a list
(simulated by an array with offset and live length)
%\footnote{
%This implies that the implementation deals with only decreasing polynomials
%in which $(B^k B)$ occurs at most 255 for each \(k\).
%We observed that it suffices to compute the \(\rho\)-property of even \(B^6 B\)
%where the number of occurrences of \((B^k B)\) never goes beyond 30 for any \(k\).
%}
whose $k$-th element stores the number of occurrences of $(B^k B)$.
For example, \((B^3 B)\circ(B^2 B)\circ(B^2 B)\circ(B^0 B)\circ(B^0 B)\) is
represented by a list \verb|[2,0,2,1]| where the 0-th element is the leftmost \verb|2|.
%That is,
%another canonical representation of the form
%\((B^n B)^{m_n}\circ\dots\circ(B^1 B)^{m_1}\circ(B^0 B)^{m_0}\) is used.
%
Since \((B^k B)^m\) is equivalent to \(B^k B^m\), 
this representation can be seen
as a variant of Curry's normal form mentioned in~\refsec{canonical}
by inserting the identity function \(B^k B^0\) for each skipped degree \(k\)
(\eg, \((B^3 B^1)\circ(B^2 B^2)\circ(B^1 B^0)\circ(B^0 B^2)\) for the above).
Using the Curry's normal form,
we can adopt a slightly-improved algorithm
by equation~\(B^{n} B^{m}\circ B^{n'} B^{m'} = 
B^{n'+m} B^{m'}\circ B^{n} B^{m}\) if \(n<n'\)
%instead of equation~\refeqn{B-swap} 
at the step (2) in \refalg{appcanon}.
%
%\begin{algorithm}[\textbf{(Application of \(B\)-terms \(P_1\) and \(P_2\) in byte-array representation)}]
%\labalg{appcanon2}~\rm
%\begin{enumerate}
%\item Build \(P_2'\) by raising each degree of $P_2$ by 1,
%\ie,~when $P_2\equiv (B^{n}\, B^{m_n})\circ\dots\circ(B^1\, B^{m_1})\circ(B^{0}\, B^{m_0})$,
%$P_2' \equiv(B^{n+1}\, B^{m_{n}})\circ\dots\circ(B^2\,B^{m_1})\circ(B^{1}\, B^{m_0})$.
%%
%\item Find a decreasing polynomial $P_{12}$ corresponding to $P_1\circ P_2'$ 
%by equation~\(B^{n} B^{m}\circ B^{n'} B^{m'} = 
%B^{n'+m} B^{m'}\circ B^{n} B^{m}\) if \(n<n'\).
%%~\refeqn{B-swap}.
%%In terms of the list representation,
%%a list $L_{12}$ is constructed
%%by appending $L_1$ and $L_2'$ and
%%repeatedly applying~\refeqn{B-swap}.
%%from its tail.
%%
%\item Obtain \(P\) by lowering each degree of \(P_{12}\) after
%eliminating the trailing 0-degree units,
%i.e.,~when $P_{12}\equiv (B^{n_1}\, B^{m_1})\circ\dots\circ(B^{n_l}\, B^{m_l})
%\circ(B^{0}\, B^{m})$ with $n_1>\dots>n_l>0$,
%$P\equiv(B^{n_1-1}\, B^{m_1})\circ\dots\circ(B^{n_l-1}\, B^{m_l})$.
%%In terms of the list representation,
%%a list $L$ is obtained from $L_{12}$
%%by decrementing each element after removing trailing 0's.
%\end{enumerate}
%\end{algorithm}
%
Regarding cycle detection of the implemenation, 
Floyd's, Brent's and Gosper's algorithms are used.
Note that the program does not terminate for the combinator
which does not have the $\rho$-property.
It will not help to decide if a combinator has the $\rho$-property.
One might observe how the terms grow by repetitive right applications
through running the program, though.

\subsection{$B$-terms not having the $\rho$-property}
\labsec{antib2n}
%!TEX root = main.tex
%\section{$\rho$-property of $B$-terms}
%\labsec{rhob}
%
%It has been conjectured~\cite{Nakano08trs} that
%``$B$-term $e$ has the $\rho$-property
%if and only if $e$ is equivalent to $B^n~B$ with some $n$''.
%In terms of list representation,
%this statement can be rephrased as
%``$B$-term $e$ has the $\rho$-property
%if and only if the list representation of $e$ has length 1''.
%
%Regarding the if-part,
%at least $B$-terms whose list representation is $B^n~B$ with $n\leq6$
%satisfies the $\rho$-property
%as demonstrated by our program mentioned in~\refsec{intro}.
%The program detects the $\rho$-property as long as the cycle exists,
%based on the Floyd's cycle-finding algorithm.
%%
%Regarding the only-if-part,
%it is impossible to check it with the program.
%For every $B$-term $e$ not equivalent to $B^n~B$ with any $n\geq0$,
%it would be necessary to know that the program \emph{does not terminate} for $e$.

A computer can check that a $B$-term has the $\rho$-property
just by calculation
but cannot show that a $B$-term does not have the $\rho$-property.
In this subsection, 
we present two methods to prove that specific $B$-terms do not have the $\rho$-property.
One employs decreasing polynomial representation as previously discussed
and the other makes use of tree grammars for binary tree representation.

\subsubsection{Using polynomial representation}
We show that $B^2$ does not have $\rho$-property as an experiment.
Note that $B^2$ has the decreasing polynomial representation
$(B^0 B)\circ(B^0 B)$ which has the smallest length, 2, among the $B$-terms that are
expected not to have the $\rho$-property.
%
%This statement may be helpful to show 
%that all $B$-terms whose decreasing polynomial representation has greater than length 1
%do not have the $\rho$-property.
%
%Since the longer polynomial is obtained as far as the longer polynomial is applied,
%the other $B$-terms that are `larger' than $B^2$ would naturally be
%expected not to have the $\rho$-property as well as $B^2$.
%We cannot present the formal proof for this implication here, though.
%Although we cannot present the formal proof of the only-if-part of the conjecture,
%we believe that list representation and \refthr{b2anti} will be useful.

To disprove the $\rho$-property of $B^2$,
we show
the following lemmas about the regularity of decreasing polynomial representation
of $\sapp{B^2}{i}$ for certain $i$.
In these statements, we use
\begin{align*}
t_m&=\frac{m^2+m}2
&&\text{and}
%\\
&
\bigodot_{i=k}^{n} f_i &=
f_{k}\circ f_{k+1}\circ f_{k+2}\circ\dots\circ f_{n-1}\circ f_{n}\text.
\end{align*}
In particular, $\bigodot_{i=k}^{n} f_i$ is an identity function if $k>n$.

\begin{lem}
\lablem{b2slope}
For any $k$ and $m$ with $0\leq k\leq m$ and $l>0$,
\begin{align}
%&
\bigodot_{i=k}^{m}(B^{m-i}\, B)^2
\circ
(B^l\,B)^2
%\nonumber\\&\qquad
&=
(B^{2m-2k+l+2}\, B)^2\circ
\bigodot_{i=k}^{m}(B^{m-i}\, B)^2
\labeqn{b2slope}
\end{align}
holds.
% where $(B^n~B)^2$ stands for $(B^n~B)\circ(B^n~B)$.
\end{lem}
\begin{proof}
This statement can be obtained
by applying equation~\refeqn{B-swap}
for $4(m-k+1)$ times.
\end{proof}

\begin{lem}
\lablem{b2general}
For any $m\geq1$ and $0\leq j\leq m$,
\begin{equation}
\sapp{B^2}{t_m+j}
=
\bigodot_{i=1}^{j} (B^{2m-i-j+2}\,B)^2
\circ
\bigodot_{i=j+1}^{m}(B^{m-i}\,B)^2
\labeqn{b2general}
\end{equation}
holds.
\end{lem}
\begin{proof}
We prove the statement by induction on $m$.
In the case of $m=1$, $t_m=1$.
When $j=0$, equation~\refeqn{b2general} is clear.
When $j=1$, equation~\refeqn{b2general} is shown by
\begin{align*}
\sapp{B^2}{2}
&=((B^0~B)\circ(B^0~B))~((B^0~B)\circ(B^0~B)) 
\\&= (B^2~B)\circ(B^2~B) = (B^2~B)^2
\end{align*}
by the application procedure over decreasing polynomial representation.

For the step case,
we show that if equation~\refeqn{b2general} holds for $m=k\geq1$
and $0\leq j\leq k$,
then it also holds for $m=k+1$ and $0\leq j\leq k+1$.
It is proved by induction on $j$ where $k$ is fixed.
When $j=0$, 
from the outer induction hypothesis, we obtain
\begin{align*}
\sapp{B^2}{t_{k+1}}
&=\sapp{B^2}{t_{k}+k+1}
\\&=\sapp{B^2}{t_{k}+k}\,B^2
\\&=
\paren{\bigodot_{i=1}^{k} (B^{2k-i-k+2}\,B)^2}
~\paren{\paren{B^0\,B}\circ\paren{B^0\,B}}
\\&=
\bigodot_{i=1}^{k} (B^{k-i+1}\,B)^2
\circ
\paren{B^0\,B}\circ\paren{B^0\,B}
\\&=
\bigodot_{i=1}^{k+1} (B^{(k+1)-i}\,B)^2
\end{align*}
by applying the application procedure over decreasing polynomial representations,
hence the statement holds for $j=0$.
When $0<j\leq k+1$, from the inner induction hypothesis and \reflem{b2slope}, we similarly obtain
\begin{align*}
%\lefteqn{
\sapp{B^2}{t_{k+1}+j}
%}\\
&=
\sapp{B^2}{t_{k+1}+j-1}\,B^2
\\&=
\paren{\bigodot_{i=1}^{j-1} (B^{2k-i-j+5}\,B)^2
\circ
\bigodot_{i=j}^{k+1}(B^{k-i+1}\,B)^2
}
%\\&\qqqqquad
\paren{\paren{B^0\,B}\circ\paren{B^0\,B}}
\\&=
\bigodot_{i=1}^{j-1} (B^{2k-i-j+4}\,B)^2
\circ
\bigodot_{i=j}^{k}(B^{k-i}\,B)^2
\circ
\paren{B^2\,B}\circ\paren{B^2\,B}
\\&=
\bigodot_{i=1}^{j-1} (B^{2k-i-j+4}\,B)^2
\circ
\paren{B^{2k-2j+4}\,B}^2
\circ
\bigodot_{i=j}^{k}(B^{k-i}\,B)^2
\\&=
\bigodot_{i=1}^{j} (B^{2(k+1)-i-j+2}\,B)^2
\circ
\bigodot_{i=j+1}^{k+1}(B^{(k+1)-i}\,B)^2\text.
\end{align*}
Therefore, the statement holds for $m=k+1$.
\end{proof}

These lemmas immediately lead to the anti-$\rho$-property of $B^2$.
\begin{thm}
\labthr{b2anti}
The $B$-term $B^2$ does not have the $\rho$-property.
\end{thm}
\begin{proof}
We prove the statement by contradiction.
If $B^2$ has the $\rho$-property,
then the set of the normal forms of $S=\{\sapp{B^2}{i}\mid i>0\}$ is finite.
Hence we can take $m$ as the maximum length of decreasing polynomial representation
among all $B$-terms in $S$.
However, decreasing polynomial representation of $\sapp{B^2}{t_{m+1}}$ has length $m+1$
according to \reflem{b2general}.
This contradicts the assumption of $m$.
\end{proof}

%\refthr{b2anti} shows just a special case of 
%the only-if-part of the conjecture,
%which claims that 
%any $B$-terms whose decreasing list representation
%has more than length 1 does not have the $\rho$-property.
%%
%The restricted result, however, may be helpful to disprove
%the $\rho$-property of the other $B$-terms.
%Note that the $B$-term $B^2$ is `smallest'
%among $B$-terms that have more than length 1
%in the sense that 
%its length is shortest and its maximum height is smallest.
%Since the longer list representation is obtained
%as far as the longer list representation is applied,
%the other $B$-terms that have more than length 1 are naturally
%expected not to have the $\rho$-property as well as $B^2$.
%Although we cannot present the formal proof of the only-if-part of the conjecture,
%we believe that list representation and \refthr{b2anti} will be useful.

\subsubsection{Using tree grammars}
%In this section, 
Another way for disproving the $\rho$-property of $B$-terms is 
to consider the $\beta\eta$-normal form of their $\lambda$-terms.
As shown in~\refsec{checkeq},
the $\beta\eta$-normal form of a $B$-term can be regarded as a binary tree.
We can disprove the $\rho$-property of $B$-terms
by observing what happens on the binary trees during the repetitive right application.
More specifically,
we first find a set which is closed under the application of a given term,
and then show the length of the spine of trees is unbounded on the repetitive right application.
This leads to the anti-$\rho$-property of the term as shown in \refthr{antibgeneral}.

%We prove that the $B$-terms $(B^k B)^{(k+2)n}$ ($k \ge 0,\ n > 0$)
%do not have the $\rho$-property.
%For example,
%$B$-term $B^2 = B~ B~ B$, which is the case of $k=0$ and $n=1$, does not have the $\rho$-property.
%%In particular, 
%To this end,
%we show that the number of variables 
%in the $\beta\eta$-normal form of $((B^k B)^{(k+2)n})_{(i)}$ is monotonically non-decreasing
%and that it implies the anti-$\rho$-property.
%Additionally, after proving that,
%we consider a sufficient condition not to have the $\rho$-property through the monotonicity.

First, we introduce some notations.
In this section, we write $\<\star,\star,\star,\ldots,\star>$ for the binary tree $\<\ldots\<\<\star,\star>,\star>,\ldots,\star>$ and identify $B$-terms with their corresponding binary trees.
For a binary tree $t=\langle\star,t_1,\dots,t_k\rangle$, we define $l(t) = \text{(the number of leaves in $t$)}$, $a(t) = k$, and $N_i(t) = t_i$ for $i=1,\dots,k$.
If $X'$ is a $B$-term, $l(X')$, $a(X')$, and $N_i(X')$ are defined to be $l(t)$, $a(t)$, and $N_i(t)$ for the binary tree $t$ corresponding to $X'$.
Suppose the $\beta\eta$-normal form of $X'$ is $\lambda x_1'\dotsLam{x'_{n'}}.\, x_1'~e_1~\dots~e_k$ and
let $X$ be another $B$-term whose $\beta\eta$-normal form is $\lambda x_1\dotsLam{x_{n}}.\, e$.
We can see $X'~ X = (\lambda x'_1\dotsLam{x'_{n'}}.\, x'_1~ e_1~ \cdots~ e_k)~ X = \lambda x'_2 \dotsLam {x'_{n'}}.\, X~ e_1~ \cdots~ e_k$
and from \reflem{bform}, its $\beta\eta$-normal form is
\[
  \begin{cases}
    \lambda x'_2\dotsLam{x'_{n'}}\dotLam{x_{k+1}}\dotsLam{x_{n}}.\, e[e_1/x_1, \ldots, e_k/x_k] & (k \le n)\\
    \lambda x'_2\dotsLam{x'_{n'}}.\, e[e_1/x_1, \dots, e_{n}/x_{n}]~ e_{n+1} ~ \cdots ~ e_k & (\text{otherwise}).
  \end{cases}
\]
Here $e[e_1/x_1,\dots,e_k/x_k]$ is the term which is obtained by substituting $e_1, \dots, e_k$ to the variables $x_1,\dots,x_k$ in $e$.

By simple computation with this fact, we get the following lemma:
\begin{lem}
\lablem{asymp}
Let $X$ and $X'$ be $B$-terms. Then
\begin{align}
  l(X'~X) &= l(X') -1 + \max\{l(X) - a(X'), \ 0\} \labeqn{leafeqn}\\
  a(X'~X) &= a(X) + a(N_1(X')) + \max\{a(X') - l(X),\  0\} \labeqn{argeqn}\\
  N_1(X'~X) &=
  \begin{cases}
    N_1(X)[N_2(X')/x_2,\ldots,N_m(X')/x_m] & (\text{if $N_1(X')$ is a leaf})\\
    N_1(N_1(X')) & (\text{otherwise})
  \end{cases} \labeqn{nexteqn}
\end{align}
where
%$m = \min\{l(N_1(X')),\  a(X)\}$.
$m = \min\{l(X),\ a(X')\}$.
\end{lem}
From this lemma, we obtain a key theorem to prove the anti-$\rho$-property.
\begin{thm}
\labthr{antibgeneral}
Let $X$ be a $B$-term and $T$ be a set of $B$-terms.
If $\left\{ X_{(i)} \ \middle|\  i \ge 1\right\} \subset T$ and
$l(X) - a(X') \ge 1$
for any $X' \in T$,
then $X$ does not have the $\rho$-property.
\end{thm}
\begin{proof}
It suffices to show the following:
Under the hypotheses of the theorem, for any $i \ge 1$, there exists $j > i$ that satisfies 
$l(X_{(j)}) > l(X_{(i)})$.
Suppose, for contradiction, that there exists $i \ge 1$ that satisfies $l(X_{(i)}) = l(X_{(j)})$ for any $j > i$.
We get $a(X_{(j)}) = l(X)-1$ by \refeqn{leafeqn} and then $a(N_1(X_{(j-1)})) = l(X)-a(X)-1$ by \refeqn{argeqn}.
Here, $l(X) - a(X) \ge 2$ since if the $\beta\eta$-normal form of $X$ is $\lambda x'_1\dotsLam{x'_{n'}}.\, x'_1~e_1~\dots~e_k$, each $e_i$ ($i=1,\dots,k-1$) has at least one variable and $e_k$ has at least two variables because otherwise the $\lambda$-term is not $\eta$-normal.
Therefore $a(N_1(X_{(j-1)})) \ge 1$, so $N_1(X_{(j-1)})$ is not a leaf for any $j>i$.
From \refeqn{nexteqn}, we obtain $N_1(X_{(j-1)}) = N_1(N_1(X_{(j-2)})) = \cdots = \underbrace{N_1(\cdots N_1(}_{j-i}X_{(i)})\cdots)$ for any $j > i$.
However, this implies that $X_{(i)}$ has infinitely many variables and it yields contradiction.
\end{proof}

Using this theorem, we prove that the $B$-term $(B^k B)^{(k+2)n}$ does not have the $\rho$-property.
The $\beta\eta$-normal form of $(B^k B)^{(k+2)n}$ is given by
\[
  \lambda x_1\dotsLam{x_{k+(k+2)n+2}}.\, x_1~ x_2~ \cdots~ x_{k+1}~ (x_{k+2}~ x_{k+3~} \cdots~ x_{k+(k+2)n+2}).
\]
This is deduced from \reflem{H-uniq} since
the binary tree corresponding to the above $\lambda$-term is
$t = \langle\underbrace{\star,\dots,\star}_{k+1},\langle\underbrace{\star,\dots,\star}_{(k+2)n+1}\rangle\rangle$
and $\mathcal{L}(t) = [\underbrace{k,\ldots,k}_{(k+2)n}]$.
In particular, we get $l((B^k B)^{(k+2)n}) = k + (k+2)n + 2$.

To apply \refthr{antibgeneral},
we introduce a set $T_{k,n}$ which satisfies the hypotheses of \refthr{antibgeneral}.
First we inductively define a set of terms $T'_{k,n}$ as follows:
\begin{enumerate}
\item $\star \in T'_{k,n}$
\item $\langle\star,~ s_1,~ \ldots,~ s_{(k+2)n}\rangle \in T'_{k,n}$ if
$s_i=\star$ for a multiple $i$ of $k+2$
and $s_i\in T'_{k,n}$ for the others.
\end{enumerate}
Then we define $T_{k,n}$ by
$T_{k,n} = \left\{ \langle t_0,~ t_1,~ \ldots,~ t_{k+1}\rangle \ \middle| \ t_0,\ t_1, \ldots, t_{k+1} \in T'_{k,n} \right\}$.
Since the binary tree of $(B^k B)^{(k+2)n}$ is $\langle\underbrace{\star,\ldots,\star}_{k+1},\langle\star,\underbrace{\star,\ldots,\star}_{(k+2)n}\rangle\rangle$,
we can see $(B^k B)^{(k+2)n} \in T_{k,n}$.
Now we shall prove that $T_{k,n}$ is closed under right application of $(B^k B)^{(k+2)n}$.
\begin{lem}
\lablem{closed}
If $X \in T_{k,n}$ then $X~ (B^k B)^{(k+2)n} \in T_{k,n}$.
\end{lem}
\begin{proof}
From the definition of $T_{k,n}$, 
if $X \in T_{k,n}$ then $X$ can be written
in the form $\langle t_0,\ t_1,\ \dots,\ t_{k+1}\rangle$
for some $t_0, \ldots, t_{k+1} \in T'_{k,n}$.
In the case where $t_0 = \star$, we have $X~ (B^k B)^{(k+2)n} = \langle t_1,~\ldots,~ t_{k+1},~ \langle\star,~\underbrace{\star,~ \ldots,~ \star}_{(k+2)n}\rangle\rangle \in T_{k,n}$.
In the case where $t_0$ has the form of 2 in the definition of $T'_{k,n}$,
then we have
$X = \langle\star,~ s_1,~ \ldots,~ s_{(k+2)n},~ t_1,~ \ldots,~ t_{k+1}\rangle$
with $s_i=\star$ for a multiple $i$ of $k+2$
and $s_i\in T'_{k,n}$ for others,
%with $s_i\in T'_{k,n}$ and so
hence
\[
  X~ (B^k B)^{(k+2)n} = \langle s_1,~ \ldots,~ s_{k+1},~ \langle s_{k+2},~ \ldots,~ s_{(k+2)n},~ t_1,~ \ldots,~ t_{k+1},~ \star\rangle\rangle.
\]
We can easily see $s_1,\ \ldots,\ s_{k+1}$, and $\langle s_{k+2},\ \ldots,\ s_{(k+2)n},\ t_1,\ \ldots,\ t_{k+1}, \ \star\rangle$ are in $T'_{k,n}$.
\end{proof}

From the definition of $T_{k,n}$,
we can compute that $a(X)$ equals $k+1$ or $(k+2)n + k + 1$
if $X\in T_{k,n}$.
Particularly, we get the following:
\begin{lem}
\lablem{targs}
For any $X \in T_{k,n}$, $a(X) \le (k+2)n + k + 1 = l((B^k B)^{(k+2)n}) - 1$.
\end{lem}

By \refthr{antibgeneral}, we get the desired result:
\begin{thm}
\labthr{bkbanti}
For any $k \ge 0$ and $n > 0$, $(B^k B)^{(k+2)n}$ does not have the $\rho$-property.
\end{thm}

We give more examples of $B$-terms which satisfy the condition in \refthr{antibgeneral} with some set $T$.
\begin{exa}
Consider $X = (B^2 B)^2 \circ (B B)^2 \circ B^2
= \<\star,\ \<\star,\ \<\star,\ \<\star,\ \star,\ \star>,\ \star>,\ \star>>$.
We inductively define $T'$ as follows:
\begin{enumerate}
\item $\star \in T'$
\item For any $t \in T'$, $\langle \star,\ t,\ \star \rangle \in T'$
\item For any $t_1, t_2 \in T'$, $\langle \star,\ t_1,\ \star,\ \langle \star,\ t_2,\ \star \rangle,\ \star \rangle \in T'$
\end{enumerate}
Then $T = \left\{ \<t_1,\ \<\star,\ t_2,\ \star>> ~\middle|~ t_1, t_2 \in T' \right\}$ satisfies the conditions in \refthr{antibgeneral}:
\begin{claim}
$\{X_{(k)}\mid k \ge 1\} \subset T$.
\end{claim}
\begin{proof}
By definition, $X \in T$.
Let $X' = \<t_1,\ \<\star,\ t_2,\ \star>>\in T$.
Then, we have
\[
X'X =
\begin{cases}
	\<\star,\ t_2,\ \star, \ \<\star,\ \<\star,\ \<\star, \ \star, \ \star>,\ \star>,\ \star>> & \text{if $t_1 = \star$}\\
	\<t_1',\  \<\star,\ \<\star,\ t_2,\ \star,\ \<\star,\ \star,\ \star>,\ \star>,\ \star>> & \text{if $t_1 = \<\star,\ t_1',\ \star>$}\\
	\<t_{11},\ \<\star,\ \<\star,\ t_{12},\ \star,\ \<\star,\ \<\star,\ t_2,\ \star>,\ \star>,\ \star>,\ \star>> & \text{if $t_1 = \<\star,\ t_{11},\ \star,\ \<\star,\ t_{12},\ \star>,\ \star>$},
\end{cases}
\]
and, in either case, $X'X \in T$.
\end{proof}
\begin{claim}
$l(X)-a(X')\ge 1$ for any $X' \in T$.
\end{claim}
\begin{proof}
  Since $a(X')$ is equal to either 1, 3, or 5, and $l(X) = 8$, $l(X) - a(X') \ge 3$.
\end{proof}
Thus, $(B^2 B)^2\circ (B B)^2 \circ B^2$ does not have the $\rho$-property.
\end{exa}
\begin{exa}
Consider $X = (BB)^3\circ B^3 = \<\star,\ \<\star,\ \star,\ \star,\ \star>,\ \star,\ \star>$.
We inductively define $T'$ as follows:
\begin{enumerate}
\item $\star \in T'$
\item For any $t \in T'$, $\<\star,\ t,\ \star,\ \star>\in T'$
\end{enumerate}
Then $T = \{\<t_1,\<\star,t_2,\star,\star>> \mid t_1,t_2 \in T'\}$ satisfies the conditions in \refthr{antibgeneral}:
\begin{claim}
  $\{X_{(k)}\mid k \ge 1\} \subset T$.
\end{claim}
\begin{proof}
  By definition, $X \in T$.
  Let $X' = \<t_1,\ \<\star,\ t_2,\ \star,\ \star>> \in T$. Then, we have
  \[
    X'X =
    \begin{cases}
      \<\star,\ t_2,\ \star,\ \star,\ \<\star,\ \<\star,\ \star,\ \star,\ \star>,\ \star,\ \star>> & \text{if $t_1 = \star$}\\
      \<t_1',\ \<\star,\ \<\star,\ \<\star,\ t_2,\ \star,\ \star>,\ \star,\ \star>,\ \star,\ \star>> & \text{if $t_1 = \<\star,\ t_1\,\ \star,\ \star>$}
    \end{cases}
  \]
  and, in either case, $X'X \in T$.
\end{proof}
\begin{claim}
  $l(X) - a(X') \ge 1$ for any $X' \in T$.
\end{claim}
\begin{proof}
  $a(X')$ equals 1 or 4 and $l(X) = 8$, so $l(X)-a(X')\ge 4$.
\end{proof}
Thus, $(BB)^3\circ B^3$ does not have the $\rho$-property.
\end{exa}

\refthr{antibgeneral} gives a possible technique to prove that $l(X_{(i)})$ diverges, or, the anti-$\rho$-property of $X$, for some $B$-term $X$.
Since the hypotheses of \refthr{antibgeneral} implies that $l(X_{(i)})$ is also monotonically non-decreasing,
we can consider another problem on $B$-terms: ``Give a necessary and sufficient condition for $l(X_{(i)})$ to be monotonically non-decreasing for a $B$-term $X$.''

%%% Local Variables:
%%% mode: latex
%%% TeX-master: "main.tex"
%%% End:

%%% Local Variables:
%%% mode: latex
%%% TeX-master: "main.tex"
%%% End:

% Possible approaches
%!TEX root = main.tex
\section{Possible approaches}
\labsec{possapp}
The present paper introduces
a canonical representation
to make equivalence check of $B$-terms easier.
The idea of the representation is based on that
we can lift all $\circ$'s (2-argument $B$)
to the outside of $B$ (1-argument $B$)
by equation~\refeqn{Bo-distr}.
One may consider it the other way around.
Using the equation,
we can lift all $B$'s (1-argument $B$)
to the outside of $\circ$ (2-argument $B$).
Then one of the arguments of $\circ$ becomes $B$.
By equation~\refeqn{Bo-push},
we can move all $B$'s right.
Thereby we find another canonical representation for $B$-terms
given by
\begin{align*}
e~ &\dcolonequals~ B ~\mid~ B~ e ~\mid~ e \circ B.
%e~ &\dcolonequals~ B ~\mid~ e \circ B ~\mid~ B~ e 
\end{align*}
We can show the uniqueness of this representation by giving a bijective transformation $f$ from it to the polynomial representation.
We define $f$ inductively by
\begin{align*}
    f(B) &= B^0B\\
    f(B~e) &= (B^{n_1+1}B)\circ\dots\circ(B^{n_k+1}B) \quad \text{if $f(e)=(B^{n_1}B)\circ\dots\circ(B^{n_k}B)$}\\
    f(e\circ B) &= f(e)\circ (B^0B).
\end{align*}
Note that $B^0B = B$ and the second rule of $f$ does not change the equivalence class of $B$-terms because $B(e_1\circ\dots\circ e_k)=(B~e_1)\circ\dots\circ(B~e_k)$ (Equation \refeqn{Bo-distr}).
We can see the inverse of this function is given by
\begin{align*}
    f^{-1}(B^0B) &= B\\
    f^{-1}((B^{n_1}B)\circ\dots\circ(B^{n_k}B)) &=
    B~(f^{-1}((B^{n_1-1}B)\circ\dots\circ(B^{n_k-1}B))) \quad (n_k>0)\\
    f^{-1}((B^{n_1}B)\circ\dots\circ(B^{n_k}B)\circ (B^0B)) &=
    (f^{-1}((B^{n_1}B)\circ\dots\circ(B^{n_k}B)))\circ B.
\end{align*}

Function application (written as $@$, explicitly) over this canonical representation
can be recursively defined by
\[
\begin{array}{r@{$\;@\;$}l@{$\;=\;$}l}
 B                & e            & B~ e \\
(e_1 \circ B)     & e_2            & e_1 ~@~ (B~ e_2)\\
(B~ e)            & B            & e \circ B \\
(B~ e_1)          & (e_2\circ B) & ((B~ e_1) ~@~ e_2) \circ B \\ % (B e1 e_2) o B
(B~ B)            & (B~ e)       & (B~ (B~ e)) \circ B \\ % B o(B e2)=B(e_1 o e_2)
(B~ (e_1\circ B)) & (B~ e_2)     & ((B~ e_1) ~@~ (B~ (B~ e_2))) \circ B \\
(B~ (B~ e_1))     & (B~ e_2)     & B~ ((B~ e_1) ~@~ e_2) 
%\\ % (B Be1)o(B e2)=B(e_1 o e_2)
%\\ % B e1 @ (B(B e2)) o B
%(B~ e_1) @ ((B~ (B~ e)) \circ B) \\ % B e1 @ (B(B e2)) o B
\text.
\end{array}
\]
Notice that the pattern matching is exhaustive.
The correctness of the equations is proved by equations~(B2') and (B3').
Termination of the recursive definition is shown 
by a simple lexicographical order of the first and the second operand of application.
Note that this canonical form can be represented by a sequence of $(B~ \square)$ and $(\square \circ B)$
where $\square$ stands for a hole.
Also, a sequence of them exactly corresponds to a single term in canonical form
by hole application.
\eg, $[(B~ \square),~ (B~ \square),~ (\square\circ B)]$ 
represents $B~ (B~ (B \circ B))$ where a nullary constructor $B$ is filled in the last element $(\square\circ B)$.
This fact may be used to find the $\rho$- or anti-$\rho$-properties.
By writing 0 and 1 for $(B~ \square)$ and $(\square\circ B)$,
the above equation can be rewritten as follows:
%\comment{If we need further discussion ...}
\begin{align*}
\varepsilon @ y &= 0y\\
1x @ y &= x@0y\\
0x @ \varepsilon &= 1x\\
0x @ 1y &= 1 (0x@y)\\
0\varepsilon @ 0y &= 100y\\
01x @ 0y &= 1 (0x@00y)\\
00x @ 0y &= 0 (0x@y)
\end{align*}
where $\varepsilon$ is used for the end marker (filling $B$ at the end).
A monomial $B$-term corresponds to a binary sequence that does not contain $1$.
If $x@y$ is always greater than $x$ in some measure when $y$ contains $1$,
we can claim the ``only-if'' part of \refcnj{B-conj}.

Waldmann~\cite{Waldmann13email} suggests
that the $\rho$-property of $B^{n} B$ may be checked
even without converting $B$-terms into canonical forms.
He simply defines $B$-terms by
\begin{align*}
  e &\dcolonequals B^k \mid e~ e
\end{align*}
and regards $B^k$ as a constant
which has a rewrite rule
$B^k~e_1~e_2~\dots~e_{k+2}\to e_1~(e_2~\dots~e_{k+2})$.
He implemented a check program in Haskell
to confirm the $\rho$-property.
Even in the restriction on rewriting,
he found that
$\sapp{(B^0 B)}{9}=\sapp{(B^0 B)}{13}$,
$\sapp{(B^1 B)}{36}=\sapp{(B^1 B)}{56}$,
$\sapp{(B^2 B)}{274}=\sapp{(B^2 B)}{310}$ and
$\sapp{(B^3 B)}{4267}=\sapp{(B^3 B)}{10063}$,
in which it requires a few more right applications
to find the $\rho$-property than the case of canonical representation.
If the $\rho$-property of $B^n B$ for any $n\geq0$ is shown
under the restricted equivalence given by the rewrite rule,
then we can conclude the ``if'' part of \refcnj{B-conj}.

%It might be possible to directly use a syntax of $\lambda$-terms
%instead of combinatory forms.
%It is not necessary to consider the canonical representation.
%This is a natural approach
%to investigate the $\rho$-property of general $\lambda$-terms.
%However, it is difficult to grasp
%how terms change by repetitive right application in general
%due to the $\beta$-reduction.

%Another possible approach is to observe the change of (principal) types
%by right repetitive application.
%Although there are many distinct $\lambda$-terms of the same type,
%we can consider a desirable subset of typed $\lambda$-terms.
%As shown by Hirokawa~\cite{Hirokawa93tcs},
%each $BCK$-term can be characterized by its type,
%that is, 
%any two $\lambda$-terms in $\CL{BCK}$ of the same principal type are identical
%up to $\beta$-equivalence.
%This approach may require observing unification between types
%in a clever way.
%%% Local Variables:
%%% mode: latex
%%% TeX-master: "main.tex"
%%% End:

% Concluding remark
%!TEX root = main.tex
\section{Concluding remark}
\labsec{concl}
%\comment{The following paragraphs are not consistent with the current contents.}
We have investigated the $\rho$-properties of $B$-terms in particular forms
so far.
\begin{table}
\caption{Summary of known results on the $\rho$-property of $B$-terms}
\labtbl{known}
\begin{tabular}{@{}l@{\hspace{4em}}l@{}}
\medskip\\[-4ex]\toprule\noalign{\medskip}
having $\rho$-property &
$B^n B$ with $0\leq n\leq 6$
\medskip\\\midrule\noalign{\medskip}
having anti-$\rho$-property &
$(B^k B)^{(k+2)n}$ with $k\geq0, n>0$
\medskip\\&$(B^2 B)^2\circ(B B)^2\circ B^2$
\medskip\\&$(B B)^3\circ B^3$
\medskip\\\bottomrule
\end{tabular}
\end{table}
\Reftbl{known} summarizes all results we obtained.
While the $B$-terms equivalent to $B^n B$ with $n\leq 6$ have the $\rho$-property,
the $B$-terms $(B^k B)^{(k+2)n}$ with $k\geq0$ and $n>0$,
$(B^2 B)^2\circ(B B)^2\circ B^2$, and
$(BB)^3\circ B^3$
do not.
We have also introduced a canonical representation of \(B\)-terms
which is useful to prove or disprove of specific \(B\)-terms.
%not equivalent to $B^n~B$ for any $n$ do not.

We introduce remaining problems related to these results.
The $\rho$-property is defined for any combinatory terms
(and closed $\lambda$-terms).
We investigated it mainly for $B$-terms 
as a simple but interesting instance
to give a partial solution of \refcnj{B-conj}
in the present paper.
%%Note that it is undecidable
%%whether a given combinatory term $X$ has the $\rho$-property
%%This is because
%%our definition of the $\rho$-property implicitly assumes 
%%the decidability of equivalence over all terms.
%%the normalizability
%%of all terms in the sequence $\{\sapp{X}{n}\}_{n=1}^{\infty}$.
%%
%%For any combinatory term $X$,
%%we can discuss decidability of the $\rho$-property of $X$-terms
%%as long as $\CL{X}$ is known normalizable.
%%Since $\CL{B}$ is normalizable,
%%We have investigated the $\rho$-property of $B$-terms
%%in the present paper
%%since they show more interesting behavior the other combinators.
%%The decidability of the $\rho$-property for $\CL{B}$ is still open, though.
%From his observation on repetitive right applications for several $B$-terms,
%Nakano~\cite{Nakano08trs} has conjectured as follows.
%%We reiterate the conjecture on $B$-terms:
%%still believes that decreasing list representation introduced
%%in the paper will be helpful to solve the problem:
%\begin{conj}
%\labcnj{B-conj}
%A $B$-term $e$ has the $\rho$-property if and only if 
%$e$ is a monomial, \ie, $e$ is equivalent to $B^n B$ with $n\geq0$.
%%That is, $\rho$-property of $B$-terms is decidable.
%\end{conj}
%The ``if'' part of \refcnj{B-conj} for $n\leq 6$ has been shown by computation
%and the ``only if'' part for $(B^k B)^{(k+2)n}~\allowbreak(k\ge0,n>0)$
%and $(B^2 B)^2\circ(B B)^2\circ B^2$
%has been shown by \refthr{antibgeneral}.
The conjecture implies that the $\rho$-property of $B$-terms is decidable.
%
%because what we need to check is whether 
%Since $BCK$-terms are also normalizable because of linearity,
One could consider the decidability of the $\rho$-property for
$BCK$ and $BCI$-terms
which is still open.
%We conjecture that the $\rho$-property of even $BCK$- and $BCI$-terms
%is decidable.
Also, the decidability for the $\rho$-property of $S$-terms and $L$-terms can
be considered.
Waldmann's work on a rational representation of normalizable $S$-terms~\cite{Waldmann00ic} may
be helpful to solve it.
%One could also consider $S$-terms 
%because normalizability of $S$-terms is decidable~\cite{Waldmann00ic}.
We expect that none of the $S$-terms have the $\rho$-property
as $S$ itself does not, though.
Regarding $L$-terms,
%
%By modifying the definition of equivalence of terms,
%we can consider different problems on the $\rho$-property.
%As mentioned in \refsec{prelim},
%Smullyan's lark combinator $L=\lambda x.\lambda y. x~(y~y)$
%does not have the $\rho$-property
%because the $\lambda$-term corresponding to $\sapp{L}{4}$
%is not normalizable.
%However, equivalence can be defined without $\lambda$-terms.
Statman's work~\cite{Statman89sc} may be helpful
where
equivalence of $L$-terms is shown decidable
up to a congruence relation induced by $L~e_1~e_2\to e_1~(e_2~e_2)$.
It would be interesting to investigate
the $\rho$-property of $L$-terms in this setting.
%\refcnj{B-conj} can be rephrased in terms of the set generated by right application,
%that is,
%``for any $B$-term $e$, the set $\{\sapp{e}{n}\mid n\geq1\}$ is finite if and only if 
%$e$ is a monomial''.
%This statement may be helpful to consider its proof for both ``if'' and ``only-if'' part.

%If all terms are normalizable,
%we can check those equivalence 
%because 
%%% Local Variables:
%%% mode: latex
%%% TeX-master: "main.tex"
%%% End:

%\section*{Acknowledgment}
%  \noindent The authors wish to acknowledge fruitful discussions with
%  A and B.

%% in general the use of bibtex is encouraged

%\begin{thebibliography}{Kos97}
%
%\bibitem[Kos97]{koslowski:mib}
%J{\"u}rgen Koslowski.
%\newblock Monads and interpolads in bicategories.
%\newblock {\em Theory Appl. Categ.}, 3(8):182--212, 1997.
%
%\end{thebibliography}

\section*{Acknowledgment}
We are grateful to Johannes Waldmann and Sebastian Maneth for their fruitful comments on this work at an earlier stage. This work was
partially supported by JSPS KAKENHI Grant Number JP17K00007.

\bibliographystyle{alpha}
\bibliography{main}

%\appendix

\end{document}